\documentclass[12pt, onecolumn]{IEEEtran}

\usepackage{mathrsfs}
\usepackage{psfrag}
\usepackage{graphicx}
\usepackage{color}
\usepackage{fancyhdr,epsfig}
\usepackage{caption}
\usepackage{subcaption}
\usepackage{cite}

\usepackage{amsthm}
\usepackage{dsfont,enumerate}
\newcommand{\ba}{\begin{align*}}
\newcommand{\eaa}{\end{align*}}

\newcommand{\nl}{\notag\\}

\newcommand{\half}{\frac 1 2}

\newcommand{\calT}{{\cal T}}
\newcommand {\bx} {\mbox{\boldmath $x$}}
\newcommand {\bX} {\mbox{\boldmath $X$}}
\newcommand {\by} {\mbox{\boldmath $y$}}

\newcommand {\bE} {\mbox{\boldmath $E$}}

\newcommand{\eqd}{\stackrel{\triangle}{=}}

\newcommand{\sbr}[1] {\left[#1\right]}
\newcommand{\cbr}[1] {\left\{#1\right\}}
\newcommand{\rbr}[1] {\left(#1\right)}
\newtheorem{theorem}{Theorem}
\newtheorem{lemma}{Lemma}
\newtheorem{corollary}{Corollary}

\newtheorem{example}{Example}
\newtheorem{proposition}{Proposition}

\def\wh{\widehat}

\def\mb{\mathbb}
\def\m{\mathcal}
\def\md{\mathds}
\def\eps{\varepsilon}
\def\tn{\textnormal}

\newcommand{\dfn}{ \stackrel{\tn{def}}{=} }

\def\vmax{v_{\textrm{max}}}

 \usepackage{amsmath, amssymb}
\usepackage{amsthm}

\newtheorem{remark}{Remark}

\allowdisplaybreaks[3]

\begin{document}

\IEEEoverridecommandlockouts

\title{Searching with Measurement Dependent Noise}

\author{Yonatan~Kaspi, Ofer~Shayevitz and Tara~Javidi\thanks{Y. Kaspi and T. Javidi are with the Information Theory and Applications (ITA) Center at the University of California, San Diego, USA. O. Shayevitz is with the Department of EE--Systems, Tel Aviv University, Tel Aviv, Israel. Emails: \{ofersha@eng.tau.ac.il, yonikaspi@gmail.com, tjavidi@ucsd.edu\}. The work of O. Shayevitz was partially supported by an ERC grant no. 639573, and a CIG grant no. 631983. Tara Javidi's work was partially supported by NSF grants CNS-1329819 and CNS-1513883. This paper was presented in part at the Information Theory Workshop 2014, Hobart, Tasmania, Australia. }}

\maketitle

\begin{abstract}
Consider a target moving at a constant velocity on a unit-circumference circle, starting at an arbitrary location. To acquire the target, any region of the circle can be probed to obtain a noisy measurement of the target's presence, where the noise level increases with the size of the probed region. We are interested in the expected time required to find the target to within some given resolution and error probability. For a known velocity, we characterize the optimal tradeoff between time and resolution, and show that in contrast to the well studied case of constant measurement noise, measurement dependent noise incurs a multiplicative gap in the targeting rate between adaptive and non-adaptive search strategies. Moreover, our adaptive strategy attains the optimal rate-reliability tradeoff. We further show that for optimal non-adaptive search, accounting for an unknown velocity incurs a factor of at least two in the targeting rate. 
\end{abstract}

\section{Introduction}
Suppose a point target is arbitrarily placed on the unit-circumference circle. The target then proceeds to move at some constant, but possibly unknown velocity $v$. An agent is interested in determining the position of the target and its velocity to within some given resolution $\delta$, with an error probability at most $\eps$, as quickly as possible. To that end, the agent can probe any region of his choosing (contiguous or non-contiguous) on the circle for the presence of the target, say once per second. The binary outcome of the probing (hit/miss) is observed by the agent through a noisy memoryless \textit{observation channel}. The agent's strategy can be either adaptive or non-adaptive; in the former case, the probed regions can depend on past observations while in the latter case the regions are set in advance. While adaptive strategies seem natural in a search setting, non-adaptive strategies are also of much interest; such strategies can be executed in parallel and therefore take significantly less time (but not less queries) than their adaptive counterparts. Our goal is to characterize the relation between $\eps$, $\delta$, and the \textit{average sample complexity}, which is the expected time $\mb{E}(\tau)$ until the agent's goal is met, for both adaptive and non-adaptive search strategies. 

Our model is unique in allowing the observation channel to depend not only on whether the target was hit or missed, but also on the \textit{size} of the probed region, where the observation channel generally becomes ``noisier'' as the size of the probed region increases. For example, the output channel can be a binary symmetric channel whose crossover probability increases as a function of the probed region size, or a binary input, additive Gaussian channel where the noise mean and variance depend on whether the target was hit or missed, as well as on the size of the probed region. This model is practically motivated if one imagines that the circle is densely covered by many small sensors; probing a region then corresponds to activating the relevant sensors and obtaining a measurement that is a function of the sum of the noisy signals from these sensors, hence the more sensors the higher the overall noise level. In fact, this model fits many practical search scenarios in which the choice of the specific physical mechanism with which one probes affects the quality and the amount of the data that is being gathered. For example, increasing the sensitivity/resolution of the acquiring device allows for better handling of noise, but at the same time increases the amount of data and hence the processing time to find the target. As a example, suppose we are searching for a small target in a vast area of the ocean (e.g., a life boat or debris) using satellite imagery. To cover such a vast area, we may choose to use relatively few images of large areas, or alternatively many ``zoomed in'' images of smaller areas. Compared to the zoomed-in images, the images covering large areas will have many artifacts that might look like the object being sought (e.g., clouds, breaking waves, unrelated debris etc.) and will therefore tend to be generally noisier than the zoomed-in images. On the other hand, acquiring enough zoomed-in images to cover a vast area can take too long to be practically feasible. There are many more such cases where the noise, accuracy and amount of required processing depend on the choice of the probing mechanism; examples include microphone sensitivity in a microphone array, ISO setting in digital imagery, biometric readings, biological tests (e.g., blood tests), DNA sequencing, etc.

When the observation channel is noiseless and the target is stationary, the optimal adaptive search algorithm is the bisection search, where at each stage we probe half of the region where the target is already known to be, and the next phase zooms in on the correct half. For a resolution $\delta$, one needs no more than $\lceil\log\frac 1 {\delta}\rceil$ queries to find the target. It is easy to show (see also the discussion that follows Theorem \ref{thrm:non-adapt}) that a non-adaptive search strategy can return the correct location with the same number of queries. Both these adaptive and non-adaptive strategies can be easily modified to find targets with known velocities. However, the classical bisection search algorithm fails completely when noise is present, since once a wrong region is selected as containing the target, the correct location will never be found. Dealing with noise in search problems has been an active research topic for many years, and some attention has been given to the fundamental limits of searching with noise from an information theoretic perspective. However, previous works focused only on stationary targets and noise models that do not depend on the search region size (\textit{measurement independent noise}). We now survey some relevant literature. 

\subsection{Previous Works}

Noisy search problems were first introduced by Renyi \cite{renyi1961problem} in the context of a two-player ``20 questions game'', in which the first player thinks of a number in some given range while the second player needs to guess this number by asking questions that have binary answers. Renyi considered the setting where the second player is known to lie with some given probability. A variation of Renyi's problem has been independently proposed by Ulam \cite{ulamadventures} (commonly known as Ulam's game) where in Ulam's setting the second player can lie adversarially, but is limited to some given number of lies. Ulam's setting and some extensions that impose various restrictions on the noise sequence have been extensively investigated in the last few decades, see \cite{pelc2002searching} and references therein. Ulam's game and its variations can also be thought of (as further discussed below) as a communication problem over a channel with noiseless feedback; this angle has been extensively studied in Berlekamp's remarkable PhD dissertation \cite{Berlekamp-Thesis}. 

In a classical work that is closer to our setting, Burnashev and Zigangirov \cite{burnashev1974interval} were the first to study a continuous search problem for a stationary target on the unit interval, from an information theoretic perspective. In their model, an agent gives the system a point on the unit interval and receives a binary answer that indicates whether the target is to the left or to the right of the given point. This binary answer might be incorrect with probability $p$, where $p$ is fixed and known (measurement independent). It is not difficult to observe (see e.g. \cite{burnashev1974interval}) that this search problem is equivalent to the problem of channel coding with noiseless feedback over a Binary Symmetric Channel (BSC) with crossover probability $p$, where the message corresponds to the target, the number of messages pertains to inverse of the resolution, and the channel noise plays the role of measurement noise; the existence of noiseless feedback pertains to the fact that the agent may use past measurements to adapt his probing strategy. The adaptive scheme used in \cite{burnashev1974interval} is based on a quantization of the Horstein scheme for communication over a BSC with crossover probability $p$ and noiseless instantaneous feedback \cite{Horstein63}, which is a generalization of the clean bisection search to the noisy case. Roughly speaking, at each instance, this scheme queries the point that equally bisects the posterior of the target location given the past observations. The optimality of the Horstein scheme for the communication problem was later established by Shayevitz and Feder \cite{shayevitz2007communication,Shayevitz11,shayevitz2016simple}, who also showed it to be an instance of the more general \textit{posterior matching} scheme. 

More recent works on the noisy 20 questions problem include the work of Jedynak et al. \cite{jedynak2012twenty}, in which a stochastic control point of view was taken in order to find the optimal (adaptive) search strategy for a target in $\mathbb{R}^n$, under the (much weaker) entropy cost constraint. This work was extended in \cite{Tsiligkaridis20Q2014} to include a multi-user collaborative scenario. A universal version of Ulam's game in its channel coding interpretation in the context of empirical capacity has been introduced Shayevitz and Feder \cite{ShayevitzFeder2005,ShayevitzFeder2009}, further studied by Eswaran \textit{et al.} \cite{EswaranEtAl2010}, and later generalized by Lomnitz and Feder \cite{LomnitzFeder2011}. Novak \cite{Nowak2011} found conditions under which generalized noisy binary search will find the target within $O(\log 1/\delta)$ samples.  In all the aforementioned works, the underlying adaptive procedure is some variation of the Horstien scheme. It was shown in \cite{Nowak2011} that the Horstein scheme achieves the optimal $O(\log 1/\delta)$ sample complexity. However, it was shown that even a repetition-code-flavored search where each question is asked multiple times and the prevailing answer is given as an input to the classical noiseless bisection search also achieves the same sample complexity. Clearly, the exact behavior of the sample complexity between the repetition and the Horstein schemes is considerably different. In contrast to these latter works, we are interested in a more refined analysis that determines not only the optimal order ($O(1/\delta)$), but also an exact characterization of the number of required queries as a function of the noise statistics, the resolution $\delta$, and the probability of error $\eps$. 

For the case of adaptive strategies, the search problem with measurement dependent noise considered in this paper can be viewed as a one-dimensional special case of the two-dimensional noisy search problem introduced and investigated by Naghshvar and Javidi \cite{NaghshvarJavidi}. Relying on the notion of \textit{Extrinsic Jensen-Shannon (EJS) divergence}, a pair of (non-asymptotic) lower and upper bounds were obtained in \cite{NaghshvarJavidi}. Furthermore,  numerical examples provided in \cite{NaghshvarJavidi} show that a fully adaptive (with exponential computational complexity) algorithm known as the \textit{MaxEJS algorithm} in many scenarios performs close to the obtained lower bound. As part of our current work, the achievability of the lower bound in \cite{NaghshvarJavidi} is established via a simple three-phase search strategy. 

Another line of work related to our search problem is on group testing. In a group testing problem, there are $K$ defective items out of a total of $N>K$ items. Any subgroup of the $N$ items can be probed for defects and a binary answer is received, indicating whether the selected subgroup contains defects or not. The seminal work by Drofmann \cite{Dorfman1943} that initiated work on this problem, has been aimed at finding syphilitic men among the millions of US draftees during World War II. Instead of taking a blood sample from each draftee and sending it to the lab, samples were mixed (grouped) together and tested. Only participants of grouped samples that tested positive were individually diagnosed. It was shown that such testing procedure can significantly reduce the total number of lab tests required to find all $K$ infected samples. A discrete version of our search problem can be viewed as a group testing problem with $K=1$ and  $N=1/\delta$. The vast majority of works on group testing deal with the noiseless setting, i.e., assume that the defectiveness test can make no errors. Recently Atia and Saligrama \cite{AtiaSaligrama2012} considered the problem of noisy group testing from an information theoretic perspective (see reference therein for related work on group testing). However, the model where the noise depends on the size of the group has never been considered; it was left as an open problem in \cite{aldridge2011interference}, see also reference therein. The measurement dependent model is extremely relevant in group testing; for instance, in the blood test example, if many samples are mixed into one vial then the amount taken from each individual is very small and the probability to detect positive samples can be much lower compared to the case where only a small number of samples are mixed together into one vial. In \cite{MeGT2015} we extend the setting of the current paper to include stationary multiple targets and measurement dependent noise. The setting considered in \cite{MeGT2015} includes group testing as a special case.

In this work we take the perspective of \cite{burnashev1974interval}, viewing the search problem as a channel coding problem. Let us take the case of a BSC observation channel as an example. When the crossover probability $p$ is measurement independent, then based on the results of \cite{burnashev_exp} for optimal error exponents for channel coding with feedback, it can be readily shown that using adaptive strategies one can achieve 
\begin{equation*}
\mb{E}(\tau) = \frac{\log{(1\slash\delta)}}{C(p)} + \frac{\log{(1\slash\eps)}}{C_1(p)} + \mathrm{O}(\log\log{\frac{1}{\delta\eps}})
\end{equation*}
where $C(p)$ is the Shannon capacity of the BSC with crossover probability $p$, and $C_1(p) = D(p\|1-p)$. This result is also the best possible up to sub-logarithmic terms. For non-adaptive strategies, given the relations between non-adaptive search and non-feedback channel coding, standard random coding exponents results \cite{GallagerBook} indicate for any fixed $0<R<C(p)$ there exists a strategy such that 
\begin{equation*}
\tau = \frac{\log{(1\slash\delta)}}{R},\qquad \log{(1\slash\eps)} = \frac{E(R,p)}{R}\cdot \log{(1\slash\delta)}
\end{equation*}
where $E(R,p)$ is the reliability function of the BSC, for which bounds are known \cite{GallagerBook}. Hence, the minimal sample complexity (with a vanishing error guarantee) is roughly the same for adaptive and non-adaptive strategies in the limit of high resolution $\delta\to 0$, and is given by $\mb{E}(\tau) \approx \frac{\log{(1\slash\delta)}}{C(p)}$. This directly corresponds to the fact that feedback does not increase the capacity of a memoryless channel \cite{shannon_zero_error}. Indeed, given the analogy between search and channel coding, it is not surprising that in all of the aforementioned works that considered both adaptive and non-adaptive strategies, there is no advantage for adaptive strategies in term of sample complexity\footnote{Note that adaptive search strategies do however exhibit superior performance over non-adaptive strategies for a fixed resolution, attaining the same error probability with a lower expected search time. They are also asymptotically better if a certain exponential decay of the error probability is desired, which directly corresponds to the fact that the Burnashev exponent \cite{burnashev_exp} exceeds the sphere packing bound \cite{GallagerBook} at all rates below capacity.}. As we shall see, this is on longer true when the crossover probability of the BSC can change as a function of the measurement, or more generality, when the observation channel is measurement-dependent. 

\subsection{Our Contributions}
This work has several contributions. We begin with stationary targets and show that, in contrast to case of measurement independent noise where adaptive and non-adaptive schemes essentially achieve the same sample complexity, when the observation channel depends on the query size then there is a multiplicative gap between the minimal sample complexity for adaptive vs. non-adaptive strategies, in the limit of high resolution. This \textit{targeting rate} gap generally depends on the variability of the observation channel noise with the size of the probed region, and can be arbitrarily large. The source of the difference lies mainly in the fact that from a channel coding perspective, the channel associated with measurement dependent noise is time-varying in quite an unusual way; it depends on the choice of the \textit{entire codebook}. The analogy we draw to channel coding allows us to treat any output alphabet and not only the binary output case. Moreover, we show that ``dithering'' allows us to consider arbitrary target location (i.e., no distribution is assumed on the target location). The maximal targeting rates achievable using adaptive and non-adaptive strategies under known velocity are given. It is shown that the optimal adaptive strategy does not need to be fully sequential (as in Horstein's scheme \cite{Horstein63}), and that a three phase search, where each phase is composed of a non-adaptive search strategy (but depends on past phases output) is enough to achieve the best possible performance. Thus, interestingly, the optimal adaptive scheme is amenable to parallelization. In addition, a rate-reliability tradeoff analysis is provided for the proposed adaptive and non-adaptive schemes. It is shown that the former attains the best possible tradeoff, i.e., achieves the Burnashev error exponent \cite{burnashev_exp}. 

Finally, we show how to extend the above results to the case of non-stationary targets. For unknown velocity, the maximal targeting rate achievable using non-adaptive schemes is shown to be reduced by a factor of at least two relative to the case of known velocity. Unlike the results for stationary targets, here, there is a gap between our direct and converse theorems. This gap diminishes as the unknown velocity becomes small. Intuitively, when the target moves at constant velocity, we need to locate it at two different time points and so that  its velocity could be calculated through the distance between the target locations. Following this intuition, it is expected that one would need at least double the samples as compared to the stationary case. However, acquiring the target is much more involved when the target velocity is unknown, and hence while the converse theorem follows the same lines as the known velocity case, the search strategy and its analysis are more intricate. If the velocity is known, it is easy to see that the stationary case results hold verbatim. 

\subsection{Organization}
The rest of this paper is organized as follows. We begin with searching for a stationary target under measurement dependent noise. In Section \ref{Sec:Prelim} we give the notation used throughout the paper and formally present the stationary target search problem. Sections \ref{Sec:NonAdapt} and \ref{Sec:STAdaptive} deal with non-adaptive and adaptive search strategies respectively. In Section \ref{Sec:MovingT}, we extend our results to include unknown target velocity. We conclude this work and discuss several interesting future directions in Section \ref{Sec:Conclusion}.

 exponent of the BSC induced by the least noisy measurement.  

\section{Preliminaries}\label{Sec:Prelim}
\subsection{Notations}
Random variables are denoted by upper case letters (e.g., $X$,$Y$) and their realizations are denoted by lower case letters. Similarly, random vectors of length $N$ are denoted by boldface letters (e.g., $\bX$ and $\bx$). Alphabets (discrete or continuous) are denoted by calligraphic letters, e.g., $\m{X}$. The Shannon entropy of a random variable (r.v.) $X$ is denoted by $H(X)$. The cardinality of a finite set $S$ is denoted by $|S|$. The Lebesgue measure of a set $S\subset\mb{R}$ is similarly denoted by $|S|$. We write $\mathds{1}(\cdot)$ for the indicator function. 
A random mapping (channel) from $X\in\m{X}$ to $Y\in\m{Y}$ is denoted by $P(y|x)$. When the channel depends on a parameter $q$, the channel is denoted by $P_q(y|x)$. For binary input channels $P_q(y|x)$, the KL divergence between the output distribution under input '1' and the output distribution under input '0' is denoted by $C_1(q)$. Namely, $$C_1(q) \triangleq  D(P_q(y|x=1)|| P_q(y|x=0)).$$ 
The mutual information between two jointly distributed r.v.s $X$ and $Y$ is denoted by $I(X;Y)$. When $X$ is binary with $\Pr(X=1)=p$ and the channel between $X$ and $Y$ depends on a parameter $q$, we denote the mutual information between $X$ and $Y$ by $I_{XY}(p,q)$. The \textit{capacity} of a binary input channel $P_q(y|x)$ is denoted by $C(q)$, i.e., $$C(q) \triangleq \max_p I_{XY}(p;q).$$

\subsection{Setup}
Let $w\in[0,1)$ be the initial position of the target, arbitrarily placed on the unit interval. 
At time $n\in\cbr{1,2,\ldots,N}$, an agent may seek the target by choosing (possibly at random) any measurable \textit{query set} $S_n\subset [0,1)$ to probe. Let $Q_n\eqd |S_n|$, namely, $Q_n$ denotes the total size of the search region at time $n$. Let $X_n = \md{1}(w_n\in S_n)$ denote the clean \textit{query output} which is binary signal indicating whether the target is in the probed region. The agent's observation at time $n$, will be denoted by $Y_n \in \cal Y$, where $\cal Y$ can be either discrete or continuous. Given $X_n=x_n$, the query observed output is governed by the \textit{observation channel} $P_{q_n}(y_n|x_n)$, which is a binary-input, $\m{Y}$-output channel that can depend on the query region size $q_n$. We will only deal with channels for which $p(y|x)$ is a continuous function of $q$ for any $0\leq q \leq 1$, $C_1(q) < \infty$\footnote{This means that the input cannot be deduced without error from the output. All the results of this paper hold trivially for clean channels.}. We will further assume that the query channel does not become worse as the search region shrinks, namely, for $q_1<q_2$, $C(q_1) \geq C(q_2)$ and $C_1(q_1) \geq C_2(q_2)$. 
We illustrate the dependence of the channel on the size of the query region in the following examples, which will serve as running examples throughout the paper. The first depicts a symmetric binary output setting while the second, a non-symmetric setting with continuous output alphabet.
\begin{example}\label{Ex:BSC}
Suppose the query output is binary. The agent obtains a corrupted version $Y_n$ of $X_n$, with noise level that corresponds to the size of the query region $q_n$. Specifically, 
\begin{align*}
Y_n = X_n + Z_n \;(\textrm{mod}\;2),  
\end{align*}
where $Z_n \sim \textrm{Bern}\left(p[q_n]\right)$, and where $p:(0,1]\mapsto [0,1/2)$ is a linear function $p=aq_n + b$. for some $a,b\geq 0$ such that $0<p<\half$. In this example, when $a>0$ the output becomes ``noisier'' as the query region size increases. When $a=0$, the query output is given by a binary symmetric channel (BSC) with crossover probability $b$. 
\end{example}

\begin{example}\label{Ex:Gaussian}
Suppose the output is Gaussian r.v. whose distribution depends both on the query size and on whether the target was hit or missed. Specifically, for constants $\mu,a\leq b$, if the target was hit ($X=1$), $Y\sim \m{N}(\mu,1+aq)$  and if the target was missed $X=0$, $Y\sim\m{N}(0,2+bq)$. This is an example where the channel behaves differently when the target is hit or missed. When hit, the signal has a constant component and is more concentrated then when missed.  
\end{example}

A \textit{search strategy} is a causal protocol for determining the sets $S_n = S_n(Y^{n-1})$, associated with a stopping time $\tau$ and estimator $\wh{W}_\tau = \wh{W}_\tau(Y^\tau)$ for the target position. A strategy is said to be \textit{non-adaptive} if the choice of the region $S_n$ is independent of $Y^{n-1}$, i.e., the sets we probe do not depend on the observations. In such a case, the stopping time is also fixed in advance. Otherwise, the strategy is said to be \textit{adaptive}, and may have a variable stopping time. A strategy is said to have \textit{search resolution} $\delta$ and error probability $\eps$ if for any $w$,
\begin{equation*}
\Pr(|\wh{W}_\tau-w| \leq \delta) \geq 1-\eps.
\end{equation*}
We are interested in the expected search time $\mb{E}(\tau)$ for such strategies, and specifically in the \textit{maximal targeting rate}, which is the maximal ratio $\frac{\log{1\slash\delta}}{\mb{E}(\tau)}$ such that $\eps\to 0$ is possible as $\delta\to 0$. The targeting rate captures the best per query exponential shrinkage of the uncertainty region when the number of queries becomes large, and plays exactly the same role as rate in a channel coding problem (this will become evident in the proof of the direct part of Theorem \ref{thrm:non-adapt}). We say that a sequence of strategies indexed by $k$ achieves a \textit{targeting rate} $R$ and an associated \textit{targeting reliability} $E = E(R)$, if $\delta_k\to 0$ as $k\to \infty$ and
\begin{equation*}
\mb{E}(\tau_k) \leq \frac{\log{(1\slash\delta_k)}}{R},\qquad \log{(1\slash\eps_k)} \geq \frac{E}{R}\cdot \log{(1\slash\delta_k)}
\end{equation*}
for all $k$ large enough. The quantity $E(R)$ is sometimes also referred to as the \textit{rate-reliability tradeoff}. 
 

\section{Stationary Target - Non-adaptive strategies} \label{Sec:NonAdapt}
We state our main result for the non-adaptive case. 
\begin{theorem}\label{thrm:non-adapt}
Let $P_q(y|x)$ be the channel between the query result and the query output for a query region of size $q$. For non-adaptive search strategies, the maximal targeting rate is given by 
  \begin{equation}\label{eq:rate_non_adaptive}
    \max_{q\in(0,\frac{1}{2})}  I_{XY}(q,q)
  \end{equation}
where $X$ is a binary r.v. with $\Pr(X=1)=q$. Moreover, for any $R$ below the maximal targeting rate, there exist a non-adaptive search strategy such that 
\begin{equation*}
\tau = \frac{\log{(1\slash\delta)}}{R},\qquad \log{(1\slash\eps)} = \frac{E_r(R,q^*)}{R}\cdot \log{(1/\delta)}
\end{equation*}
where $q^*$ is the maximizer in \eqref{eq:rate_non_adaptive}, and 
\begin{equation*}
  E_r(R,q^*) = \max_{\rho\in (0,1)} E_0(\rho,q^*) - \rho R
\end{equation*}
is the random coding exponent \cite{GallagerBook} for the channel  $P_{q^*}(Y|X)$ with input distribution $q^*$, at rate $R$. 
\end{theorem}

Theorem \ref{thrm:non-adapt} clearly bears many similarities to the standard channel coding theorem, with some (important) differences. We will start by building intuition that will clarify the connection to channel coding, and then highlight the distinctions. This discussion will be followed by formal proofs of the direct and converse parts of Theorem 1.

\subsection{Discussion and Intuition}\label{Sec:Discussion}
We begin by adapting the result of the Theorem to Example \ref{Ex:BSC}. The mutual information in the binary case is given by 
\begin{align}
I_{XY}(q;q) = h_2(q*p(q))-h_2(p(q)). \label{eq:Ipq}
\end{align}
In contrast to the mutual information expression between input and output of a BSC where the noise is given and does not depend on the transmission scheme, here the second term in \eqref{eq:Ipq} depends on the input prior as well. We plot the dependence of \eqref{eq:Ipq} on $q$ for our specific choice of $p(q)$ which is given by a linear function of $q$ in Fig. \ref{Fig:LinIpq}. As can be seen, our mutual information functional is no longer concave in the input distribution as in the classic channel coding setup, due to the effect of measurement dependence.
We also numerically plot $I_{XY}(q;q)$ for the Gaussian case of Example \ref{Ex:Gaussian} in Fig. \ref{Fig:GauusianIq}\footnote{Here, the output, $Y$, is a mixture of two Gaussians with weights $q,1-q$. Since there is no closed form expression for the differential entropy of a Gaussian mixture, we approximate it numerically}. As can be seen, when the dependence of the variance of the output on the query size is strong, the optimal query size becomes small. Also, since in this example the signal is more concentrated when the target is hit, the optimal query size for weaker dependence contains more than half of the queried region.
\begin{figure} 
\centering 
\begin{subfigure}{0.45\textwidth}
 \includegraphics[width=\textwidth, height=0.2\textheight]{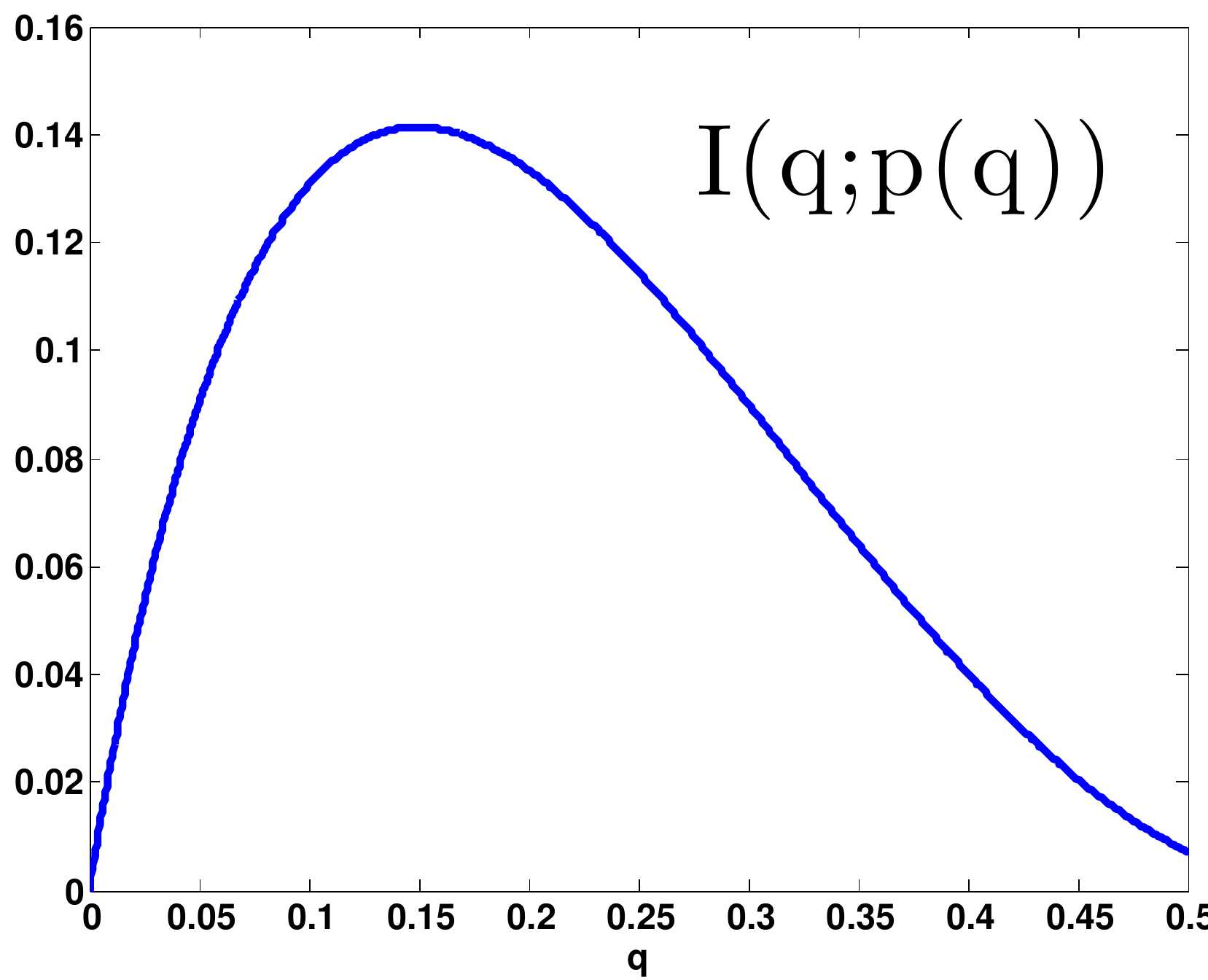}
 \caption{$I_{XY}(q;q)$ as a function of $q$ for Example \ref{Ex:BSC} with $a=0.7$ and $b=0.1$}
 \label{Fig:LinIpq}
 \end{subfigure}
\begin{subfigure}{0.45\textwidth}
 \includegraphics[width=\textwidth, height=0.2\textheight]{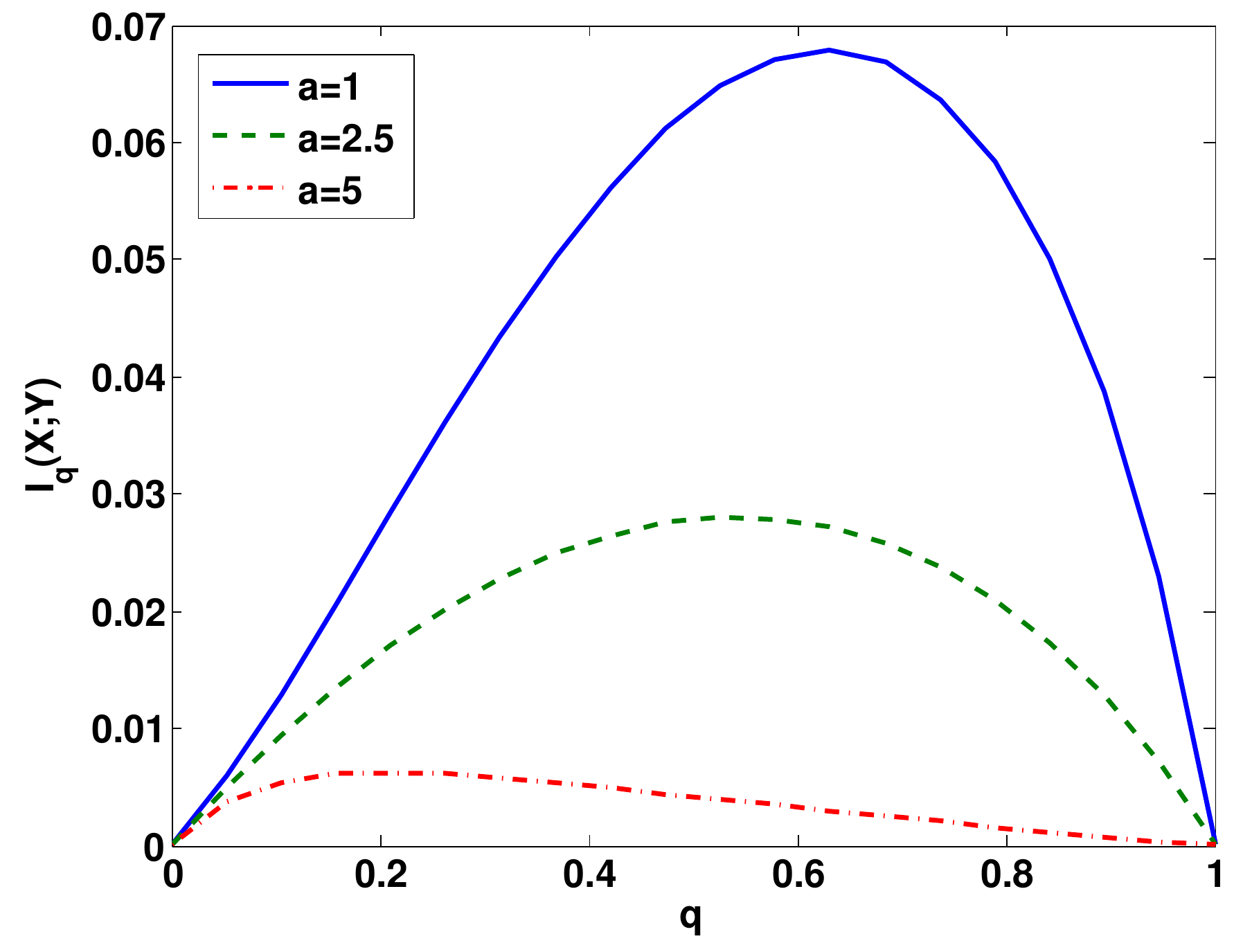}
 \caption{$I_{XY}(q;q)$ as a function of $q$ for Example \ref{Ex:Gaussian} with $\mu=0.1$, $b=5$ and various values of $a$}\label{Fig:GauusianIq}
\end{subfigure}
\caption{ Behavior of $I_{XY}(q;q)$}
\end{figure}

To build intuition, let us first consider a simple noiseless model, where we see the binary noiseless query result as the query output. We divide the unit interval into $M=1/\delta$ small equi-sized ``sensors'' and use a binary query matrix (codebook) with $M$ rows and $N$ columns, where $N$ will be the total number of queries we will execute. Each row in the codebook represents a specific sensor and each column represents a specific query. At time $n$, the query region $S_n$ is the union of all the sensors that have '1' in the $n$-th column. This is a non-adaptive search since the matrix is fixed in advance (does not depend on $Y^n$). The simple noiseless model is depicted in Fig. \ref{Fig:Clean}. The query result is '1' whenever the sensor covering the target location is activated, and '0' otherwise. Therefore, in the clean model we see at the output the exact row that corresponds to the correct sensor. The decoder will make an error in this case only if there is another row, belonging to another sensor, with exactly the same binary values as the correct row. In order to keep all rows different, $N=\lceil \log M \rceil$ is enough. Such $N$ corresponds to a maximal targeting rate of $R=1$, which is the best we could hope for. Note that the well known binary ``divide-and-conquer'' adaptive search algorithm, which reduces the query region by half at each stage, achieves the same $N=\lceil \log M \rceil$ to produce an accurate estimate. 

Moving on to noisy settings, let us start by adding a query independent channel between the query result and the output (see Fig. \ref{Fig:IndNoise}). As before, the query result is the codeword pertaining to the correct sensor. The decoder in this setting faces the exact same problem as a decoder in a standard communication setting, where it needs to decide on the correct transmitted codeword given the channel output. The maximal targeting rate will therefore be achieved by a codebook that achieves the capacity of the binary input channel $P(y|x)$. In our example, where the channel is a BSC(p), the maximal rate is given by (formally proved later) the well known capacity of the BSC(p), $1-h_2(p)$. This same result was obtained in \cite{burnashev1974interval} with an adaptive strategy. 

In the last two examples, adaptive and non adaptive search strategies achieve the same maximal targeting rate. Given the correspondence seen in these examples to the problem of single user channel coding, this is not surprising since it is well known that feedback does not increase the capacity of a single user channel. As seen in the proof of the converse to Theorem \ref{thrm:non-adapt}, adaptive schemes cannot outperform their non-adaptive counterparts in terms of rate as long as the channel does not depend on the query size. In both the above examples, when the channel is a BSC, the optimal codebook will have (approximately) the same number of 1's and 0's in each row and column (as the capacity achieving prior is $\textrm{Ber}(1/2)$), which in turn implies that the optimal query size in each query is $1/2$.

When the observation channel depends on the query size, the search problem is no longer equivalent to any classic channel coding problem. This happens since the channel is no longer fixed, but rather time-varying and determined by the \textit{choice of the codebook}, as the latter determines the query sizes. This dependence is depicted in Fig. \ref{Fig:DepNoise}. As seen in Fig. \ref{Fig:LinIpq}, when the channel depends on the query size, the optimal query size is far from $1/2$. In Example \ref{Ex:BSC}, the optimal solution backs off from the uniform distribution (which would be optimal for any BSC) in order to communicate through a better channel, albeit with a suboptimal input distribution.

\begin{figure}[h!]
\centering 
\begin{subfigure}[b]{0.7\textwidth}
 \includegraphics[width=\textwidth, height=0.15\textheight]{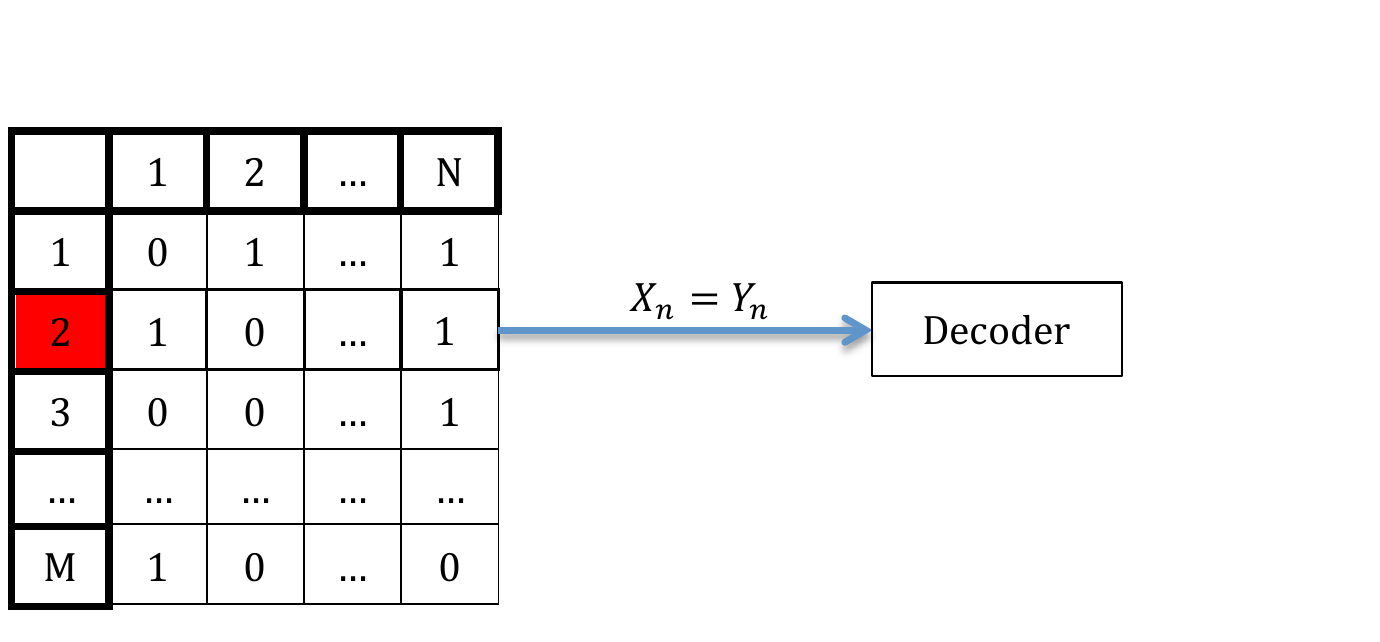}
 \caption{Querying with a clean channel}\label{Fig:Clean}
\end{subfigure}\\
\begin{subfigure}[b]{0.7\textwidth}
 \includegraphics[width=\textwidth, height=0.15\textheight]{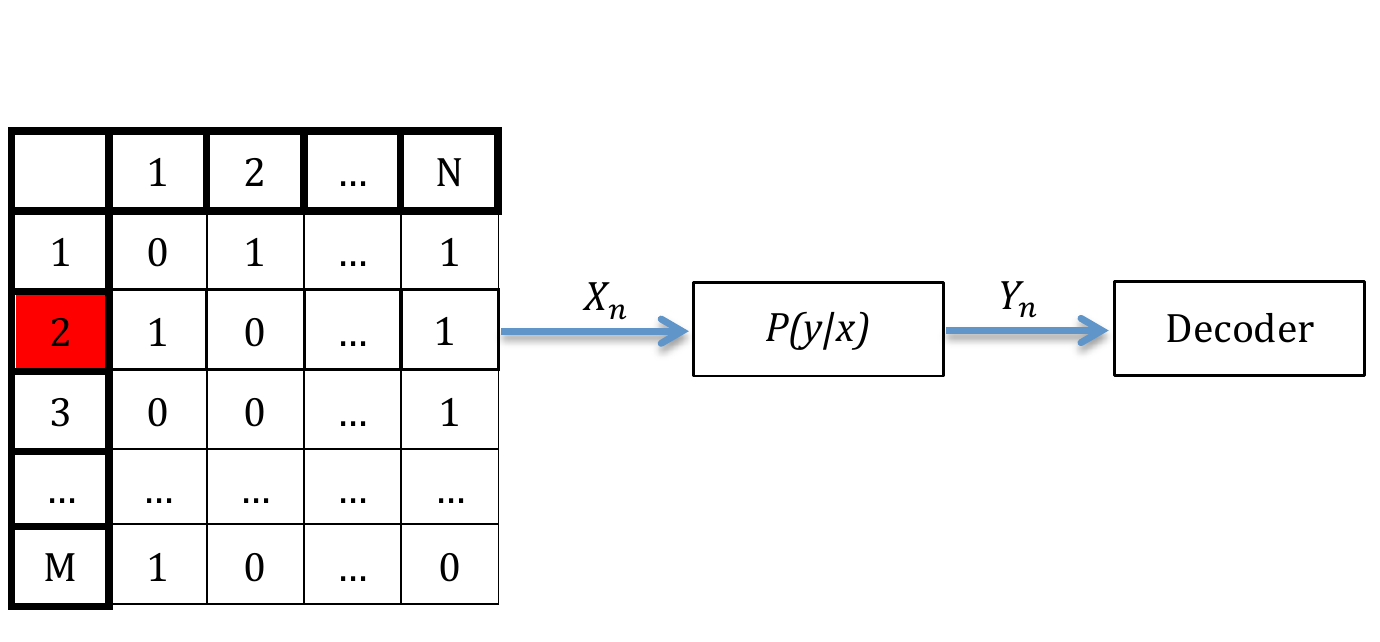}
 \caption{Querying through a measurement independent channel}\label{Fig:IndNoise}
\end{subfigure}\\
\begin{subfigure}[b]{0.7\textwidth}
 \includegraphics[width=\textwidth, height=0.15\textheight]{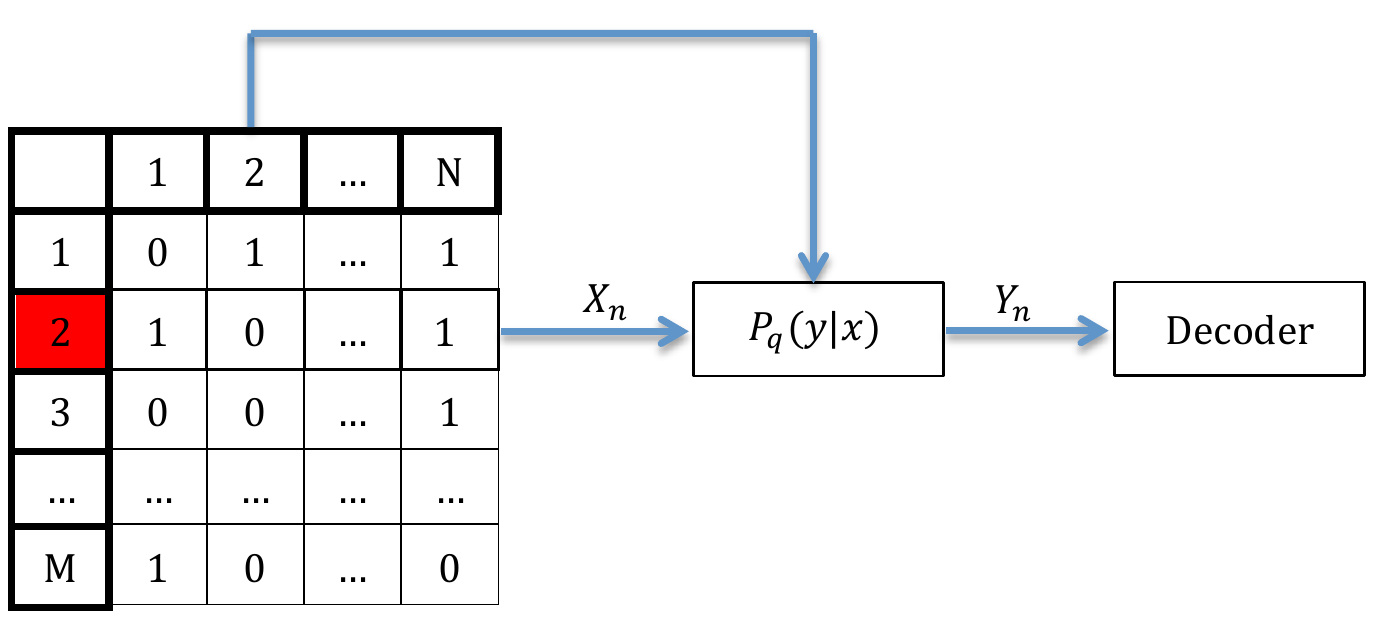}
 \caption{querying through a measurement dependent channel}\label{Fig:DepNoise}
\end{subfigure}
\caption{}
\end{figure}

\subsection{Proof of Achievability}\label{Sec:Direct}
Achievability is obtained via random coding using the query scheme described in the previous subsection. 
Formally, let $q^*$ denote the input distribution that achieves the supremum in \eqref{eq:rate_non_adaptive}. We partition the unit interval into $M = \frac{1}{\delta}$ equi-sized subintervals $\{b_m\}$. Set $R = \frac {\log M} N$.
%
%
We now draw a codebook with $M$ rows, where each row $\bx_m$ has $N$ bits. The codebook is drawn i.i.d. $\textrm{Bern}(q^*)$. We define our random query set $S_n$ according to the codebook's columns: \footnote{The summation and modulo operations are taken as operations on sets, in the usual Minkowski sense.}  
\begin{align*}
  S_n \dfn \rbr{A + \bigcup_{m:x_{m,n}=1}b_m} \mod 1 
\end{align*}
where $A\sim\textrm{Unif}([0,1))$ is a random ``dither'' signal, independent of the codebook. This dithering procedure renders our setting equivalent to the setup where the initial position is uniform and independent, and where the query sets are given by $A=0$. Thus, without loss of generality, we proceed under this latter setup. 

Unlike the standard analysis of random coding \cite{cover}, here, the channel depends on the codebook through the size of the query region. Therefore, since $S_n$ is a random variable, so is the choice of the channel. The following lemma states that for a random codebook, and $N$ large enough, the channels used throughout the $N$ queries will be arbitrarily close to $P_{q^*}(y|x)$ with probability that converges to unity double exponentially fast. 
  
\begin{lemma}\label{Lem:DEFvanish}
Let $\m{A}$ be the event where $||S_n|-q^*|\leq \epsilon$ for all $n$. Then for any $\epsilon > 0$, 
\begin{align}
\Pr(\m{A}^c) \leq N2^{-2^{NRD(q^*+\epsilon||q^*)}},  
\end{align}
where $D(q^*+\epsilon||q^*)=(q^*+\epsilon)\log\frac{q^*+\epsilon}{q^*} + (1-(q^*+\epsilon))\log\frac{1-(q^*+\epsilon)}{1-q^*}$.
\end{lemma}
\begin{proof}
Note that $|S_n| =2^{-NR} \sum_{i=1}^{2^{NR}} \mathds{1}(X_{i,n}=1)$. For the $n$th column, $\Pr(||S_n|-q^*|\geq \epsilon)\leq 2^{-2^{NR D(q^*+\epsilon||q^*)}}$ follows directly from Chernoff's bound \cite{cover}. Then, the union bound is applied to accommodate all $N$ columns.
\end{proof}
 
\begin{remark}\label{remark1}
Under the event $\m{A}$ in Lemma \ref{Lem:DEFvanish}, and using the continuity of the observation channel in $q$, the probability that the measurements are observed through a channel which is not arbitrarily close to $p_{q^*}(y|x)$ vanishes double exponentially fast as $N\to\infty$. Therefore, in the derivation below we will assume that the observation channel is $p_{q^*}(y|x)$. While there could be a mismatch between the decoder which decodes according to $p_{q^*}(y|x)$ and the actual channel, this mismatch is arbitrarily small. The continuity of the channel and hence the error exponent in $q$, along with Lemma \ref{Lem:DEFvanish} assures us that the error exponent which pertains to the mismatched decoder and the actual channel is arbitrarily close to the one we derive assuming that the observation channel is $p_{q^*}(y|x)$. 
\end{remark}

After $N$ queries, we find the codeword that has the highest likelihood under the assumption that the measurements are observed through $P_{q^*}(y|x)$. We now analyze the probability of error. This analysis follows Gallager's random error exponent analysis \cite{GallagerBook} with the distinction that under the conditioning on $\m{A}$, the codewords are not independent.
 
We write the average probability of error as 
\begin{align}
\overline{P}_e = \Pr(\m{A})\Pr(e|\m{A}) + \Pr(\m{A}^c)\Pr(e|\m{A}^c)\label{eq:ExpInit0}.  
\end{align}
The second term vanishes double exponentially fast. For the first term we have
\begin{align}
 	\Pr(e|\m{A})=\int_{\by}d\by\sum_{\bx_k}\Pr(\bx_k|\m{A})P_{\m{A}}(\by|\bx_k)\Pr(e|\bx_k,\by,\m{A}), \label{eq:ExpInit}
\end{align}
where $\by$ are the observations and $P_{\m{A}}(\by|\bx_m)$ is the channel (probability density function) induced by the event $\m{A}$. In what follows, as explained in Remark \ref{remark1}, we assume that this channel is $P_{q^*}(\by|\bx)$. 

Let $\m{E}_{k'}$ denote the event that the  codeword $\bx_{k'}$ is chosen instead of $\bx_k$. 
Following \cite{GallagerBook} we write for $0\leq \rho \leq 1$
\begin{align}\label{eq:Tk}
 	\Pr(e|\bx_k,\by,\m{A}) \leq \rbr{\sum_{k'\neq k}\Pr(\m{E}_{k'}|\m{A})}^{\rho}
\end{align}
and 
\begin{align}
 	\Pr(\m{E}_{k'}|\m{A}) = \sum_{\bx_{k'}: P_{\m{A}}(\by|\bx_k)\leq  P_{\m{A}}(\by|\bx_{k'})} \Pr(\bx_{k'}|\bx_k,\m{A})\label{eq:CondErr}
\end{align}
Note that unlike \cite[eq. 5.6.8]{GallagerBook}, we cannot assume the codewords are independent under event $\m{A}$. To overcome this, we note that by simple application of the law of total probability we have:
\begin{align}
	\Pr(\bx_k, \bx_{k'}|\m{A}))&\leq \frac {\Pr(\bx_k,\bx_{k'})}{1-\Pr(\m{A}^c)}  = \frac {Q(\bx_k)Q(\bx_{k'})}{1-\Pr(\m{A}^c)} \label{eq:TotProb1}\\
	\Pr(\bx_{k}|\m{A}) &\geq Q(\bx_{k}) - \Pr(\m{A}^c) \label{eq:TotProb2}
\end{align}
where $Q(\cdot)$ denotes the random coding prior. Using \eqref{eq:TotProb1} and \eqref{eq:TotProb2} along with the Bayes rule we have
\begin{align}
	\Pr(\bx_{k'}|\bx_k,\m{A}) &= \frac{\Pr(\bx_{k'},\bx_k|\m{A})}{\Pr(\bx_{k}|\m{A})}\nl
	&\leq \frac {Q(\bx_k)Q(\bx_{k'})} {(1-\Pr(\m{A}^c))(Q(\bx_{k}) - \Pr(\m{A}^c))}\nl
	&\leq \frac {Q(\bx_k)Q(\bx_{k'})} {Q(\bx_{k})(1-\Pr(\m{A}^c))(1 - \Pr(\m{A}^c)/Q(\bx_{k}))}\nl
	&\leq \frac {Q(\bx_{k'})} {(1-\Pr(\m{A}^c)/Q_{min}^N)^2}\nl
	&= Q(\bx_{k'})(1+\epsilon') \label{eq:CondisUncond}
\end{align}
where $Q_{min}$ denotes the probability of the least probable binary symbol under Q. 
Note that $Pr(\m{A}^c)/Q_{min}^N$ vanishes double exponentially fast due to Lemma \ref{Lem:DEFvanish} and therefore $\epsilon'$ in \eqref{eq:CondisUncond} also vanishes double exponentially fast. We conclude that while the conditioning on $\m{A}$ introduces dependencies between the codewords, since we condition on a sufficiently high probability event, the probability of the conditioned event will become arbitrarily close to the probability of the unconditioned event. 

After substituting \eqref{eq:CondisUncond} in \eqref{eq:CondErr} and \eqref{eq:Tk} we can follow Gallager's derivation of the random error exponent \cite{GallagerBook} to arrive at
\begin{align}
	\Pr(e|\bx_k,\by,\m{A}) \leq \sbr{(M-1)\sum_{\bx}Q(\bx)(1+\epsilon')\frac {P_{q^*}(\by|\bx)^{\frac 1 {1+\rho}}} {P_{q^*}(\by|\bx_k)^{\frac 1 {1+\rho}}} }^{\rho} \label{eq:ExpEgxyA}
\end{align}
Substituting \eqref{eq:ExpEgxyA} into \eqref{eq:ExpInit} we arrive at
\begin{align}
	\Pr(e|\m{A})\leq  M^{\rho}\int_{\by}d\by\sum_{\bx_k}Q(\bx_k)(1+\epsilon')P_{q^*}(\by|\bx_k)^{\frac 1 {1+\rho}}\sbr{\sum_{\bx}Q(\bx)(1+\epsilon')P_{q^*}(\by|\bx)^{\frac 1 {1+\rho}}  }^{\rho}
\end{align}
Now, we can harness the fact $\epsilon'$ vanishes double exponentially fast and therefore does not affect the error exponent, hence 
\begin{align}
	\Pr(e|\m{A})\stackrel{\cdot}{\leq}  M^{\rho}\int_{\by}d\by\sum_{\bx_k}Q(\bx_k)P_{q^*}(\by|\bx_k)^{\frac 1 {1+\rho}}\sbr{\sum_{\bx}Q(\bx)P_{q^*}(\by|\bx)^{\frac 1 {1+\rho}}  }^{\rho}
\end{align}
From this point on, since the channel is memoryless and the codewords under $Q(\cdot)$ are drawn i.i.d. we can follow \cite[5.6.11-5.6.12]{GallagerBook} verbatim to arrive at
\begin{align}
	\Pr(e|\m{A})&\stackrel{\cdot}{\leq}  2^{-N(E_0(\rho,q^*) - \rho R)},\nl
	E_0(\rho,q^*)&= -\log\int_{y}dy\cbr{q^*P_{q^*}(y|x=1)^{\frac 1 {1+\rho}}+(1-q^*)P_{q^*}(y|x=0)^{\frac 1 {1+\rho}}}^{1+\rho}
\end{align}
The exponent is positive as long as $R\leq I_{XY}(q^*,q^*)$. Since the second addend in \eqref{eq:ExpInit0} vanishes double exponentially fast, this concludes the achievability part of the proof of Theorem \ref{thrm:non-adapt}.

\subsection{Proof of Converse}
Denote the fixed stopping time by $\tau = N$. Let  $\{S_n\}_{n=1}^N$ be any non-adaptive strategy achieving an error probability $\eps$ with search resolution $\delta$. We prove the converse holds even under the less stringent requirement where the initial position is uniformly distributed $W\sim\textrm{Unif}[0,1)$. The idea of the converse proof is to discretize the problem, and then follow the standard channel coding converse.

Partition the unit interval into $\lceil\beta/\delta\rceil$ equi-sized intervals for some constant $\beta\in(0,\tfrac{1}{2})$, and let $W'$ be the index of the interval containing $W$. Any scheme $\{S_n\}$ that returns $W$ with resolution $\delta$ and error probability $\epsilon$, can be made to return $W'$ with an error probability at most $\eps'\dfn \eps + 2\beta$, where the latter addend stems from the probability that $W$ is close to a boundary point of the grid, and $\wh{W}$ is estimated to be in the adjacent (wrong) bin. 

When the target is uniformly distributed on the unit interval, for any given search region, the probability that we hit the target depends only on the size of the region. Therefore, we have that $X_n\sim\mathrm{Bern}(q_n)$ where $q_n\dfn |S_n|$ and that $Y_n$ is obtained from $X_n$ through a memoryless time-varying channel $P_{q_n}(y|x)$. Since the scheme is non-adaptive and fixed in advance, the channels $\cbr{P_{q_n}(y|x)}_{n=1}^N$, while time varying, are known in advance, and given the non adaptive search strategy, do not depend on past outputs $Y^{n-1}$. Following the steps of the converse to the channel coding theorem, we have
\begin{align}\label{eq:converse}
\log\left\lceil \frac{\beta}{\delta}\right\rceil &= H(W') \notag\\ 
&= I(W';Y^N) + H(W'|Y^N) \notag\\ 
&\stackrel{(\textrm{a})}{\leq} I(W';Y^N) + N\eps' \notag\\ 
&= \sum_{n=1}^N I(W';Y_n|Y^{n-1}) + N\eps' \notag\\ 
&\leq \sum_{n=1}^N I(W',Y^{n-1};Y_n) + N\eps' \notag\\ 
&\stackrel{(\textrm{b})}{\leq}  \sum_{n=1}^N I(W',W,Y^{n-1};Y_n) + N\eps' \notag\\ 
&\stackrel{(\textrm{c})}{=}  \sum_{n=1}^N I(W',W, X_n, Y^{n-1};Y_n) + N\eps' \notag\\ 
&\stackrel{(\textrm{d})}{=} \sum_{n=1}^N I(X_n;Y_n) + N\eps' \notag\\ 
&\stackrel{(\textrm{e})}{=} \sum_{n=1}^N I_{XY}(q_n;q_n) + N\eps' \notag\\ 
&\leq N\sup_q I_{XY}(;qq) + N\eps', 
\end{align}
where (a) is by virtue of Fano's inequality, (b) follows since conditioning reduced entropy, (c) follows since $X_n$ is a function of $W$, (d) follows since the search scheme does not depend on $Y^{n-1}$ and the channel from $X_n$ to $Y_n$, given the scheme, is memoryless. Finally, (e) follows from the definition of $I_{XY}(p;q)$ and noticing that  the size of the search region governs the input distribution. Dividing by $N$ we obtain 
\begin{IEEEeqnarray*}{rCl}
R &\leq&  \sup_q I_{XY}(q;q)  -\frac{\log\beta}{N}+ \eps + 2\beta.
\end{IEEEeqnarray*}
Noting that the inequality above holds for any $\beta\in(0,\tfrac 1 2)$, the converse now follows by taking the limit $N\to\infty$, then $\beta\to 0$ and then requiring $\eps\to 0$. 

\begin{remark}
Note that here, in contrast to the standard memoryless channel coding setup where the channel noise is strategy independent, (d) above does not generally hold when an adaptive strategy (i.e., feedback) is employed; this stems from the fact that in this case, the channel itself would generally depend on $Y^{n-1}$ (since $|S_n|$ is determined by it), and therefore $Y^{n-1}\to X_n\to Y_n$ would not form a Markov chain. This means that when adaptive strategies are considered and the channel $P_{q}(y|x)$ is not constant in $q$, the above converse will not hold. As will be seen in Section \ref{Sec:STAdaptive}, indeed, the targeting rates that are attainable with adaptive schemes can significantly exceed their non-adaptive counterparts. However, if $P_{q}(y|x)$ does not depend on $q$, the above converse proof holds for adaptive schemes as well, meaning that in that case there is no gain in terms of rate from adaptive schemes (as is well known).  
\end{remark}

\section{Stationary Target - Adaptive strategies}\label{Sec:STAdaptive}
In this section, we consider the gain that can be reaped by allowing the search decisions to be made adaptively. Here, the duration of search $\tau$ will generally be a random stopping time dependent on the measurements sample path. Moreover, the choice of the probing regions $S_n$, for $n$ up to the horizon $\tau$, can now depend on past measurements. We characterize the gain in terms of the maximal targeting rate, and the targeting rate-reliability tradeoff. While the most general adaptive schemes can be fully sequential (as for example in the classical case of binary search over noiseless channels), we will see that adaptive schemes with very few phases, where each phase employs a non-adaptive strategy but depends on the outcome of the previous phase (hence overall adaptive) are enough to achieve the optimal rate-reliability tradeoff. 
Furthermore, we will show that adaptivity allows us to achieve the maximal possible rate and reliability, i.e., those associated with the best observation channel in our model. 

In the classical channel coding setup it is well known that feedback, while not affecting the capacity, can significantly improve the targeting reliability (error exponent). Following this, in Section \ref{Sec:Validation} we present simple adaptations of two well known schemes to our setting, that increase the targeting reliability beyond the sphere-packing bound. In Section \ref{Sec:2PhaseAndValidation} we then proceed to harness the fact that for measurement dependent observation channels, using feedback does not only improves the targeting reliability, but also increases the maximal targeting rate. Note again that this fact does not contradict Shannon's famous result \cite{Shannon56} that feedback does not increase capacity in a (classical) point to point channel. 

\subsection{Non Adaptive Search with Validation}\label{Sec:Validation}
As a first attempt at an adaptive strategy, we continue with the non-adaptive search from the previous section, but allow the agent to validate the outcome of the search phase. If the validation fails, the agent restarts the search and discards all past data. We call such a validation failure an \textit{erasure} event. Therefore, if the probability of erasure is given by $\Pr(\m{E})$ and $N$ queries were made until either the location is estimated or erasure is declared, the average number of queries will be given by 
\begin{align}
	\bE[\tau] = N\rbr{1+\Pr(\m{E})+\Pr(\m{E})^2+\Pr(\m{E})^3 \ldots} = \frac N {1-\Pr(\m{E})}, 
\end{align}
incurring a rate penalty of $1-\Pr(\m{E})$ compared to the non adaptive scheme. Therefore, as long as $\Pr(\m{E})$ is kept small, there is effectively no penalty in rate. However, given that erasure was not declared, the probability of error can be greatly improved compared to the non-adaptive scheme.   
We will consider two validation schemes, due to Forney  \cite{Forney1968} and Yamamoto-Itoh \cite{YamaItoh1980}. 

In \cite{Forney1968}, Forney considered a communication system in which a decoder, at the end of the transmission, can signal the encoder to either repeat the message or continue to the next one. Namely, it is assumed that a one bit ``decision feedback'' can be sent back to the transmitter at the end of each message block. This is achieved by adding an erasure option to the decision regions, that allows the decoder/agent to request a ``retransmission'' if uncertainty is too high, i.e., to restart the exact same coding process from scratch. In our search problem, this means that if the likelyhood of one of the codewords is not high enough compared to all other codewords, erasure is declared and the search is restarted from scratch.
\begin{proposition}\label{prop:Forney}
Let $P_q(y|x)$ be the channel between the query result and the query output for a query region of size $q$. The targeting rate-reliability tradeoff for non-adaptive scheme with Forney's decision-feedback validation is given by 
\begin{align*}
	\nonumber E =  \max_{\rho\geq 1}E_0(R,\rho)-\rho R,
\end{align*}
with	
\begin{align*}
	E_0(\rho,R) = \int_y dy\sum_{x=0}^1 Q^*(x)P_{q^*}(y|x)\sbr{\log P_{q^*}(y|x) - \log\rbr{\sum_{x'=0}^1Q^*(x')P_{q^*}(y|x')^{\frac 1 {\rho}}}^{\rho}}
\end{align*}
where $q^*$ is the maximizer in \eqref{eq:rate_non_adaptive} and $Q^*(x=1) = q^*$.
\end{proposition}
The analysis and proof of this proposition is given in \cite[p. 213]{Forney1968} and therefore we only sketch the idea here. 
Given $\by$, a codeword $k$ will be declared as the output if $\frac {P_{q^*}(\by|\bx_k)} {\sum_{k'\neq k}P_{q^*}(\by|\bx_k')}\geq 2^{NT}$, where $T>0$ governs the tradeoff between the probability of error and the probability of erasure.  
Let $\varepsilon_1$ denote the event where $\by$ does not fall within the decision region of the correct codeword and let $\varepsilon_2$ denote the event of undetected error. This gives us $\Pr(\m{E}) = \Pr(\varepsilon_1) - \Pr(\varepsilon_2) \leq \Pr(\varepsilon_1)$. Forney showed that the choice of $T$ controls the tradeoff between $\Pr(\varepsilon_1)$ and $\Pr(\varepsilon_2)$ and that we can keep $\Pr(\varepsilon_1)$ arbitrarily small (vanishing with a positive exponent) while achieving the exponential decay rate given by Proposition \ref{prop:Forney} for $\Pr(\varepsilon_2)$.
\\
The second validation scheme we consider was proposed by Yamamoto and Itoh in \cite{YamaItoh1980} in the context of channel coding with clean feedback. Unlike Forney's scheme which requires only one bit of feedback, this scheme requires the decoder to feed back its decision. While perfect feedback is impractical in a communication system, in our model it is inherent (as the encoder and decoder are the same entity) and can be readily harnessed. Let $0\leq \lambda\leq 1$. A block of $N$ queries is divided into two parts as follows. In the first part, of size $N(1-\lambda)$, the non-adaptive scheme of Section \ref{Sec:Direct} is used. After completing the search phase with resolution $\delta$, the agent continues to probe the estimated target location, namely an interval of size $\delta$ with $\lambda N$ queries. If the probed region contains the target, the output of the validation phase should look like a sequence of '1's passing through $P_{\delta}(y|x=1)$, otherwise (the probed region does not contain the target), the output is drawn according to $P_{\delta}(y|x=0)$. Thus, after the probing phase, the agent is faced with a binary hypothesis test with $\lambda N$ i.i.d samples drawn according to $P_{\delta}(y|x=1)$ if the target was hit ($H_1$) and according to $P_{\delta}(y|x=0)$ otherwise ($H_0$). If the agent declares $H_0$ (erasure), the search is restarted from scratch. If the agent declares $H_1$, the target location given by the first phase is declared as the estimated target location. Thus, an error will occur if the agent declares $H_1$ when the target was missed in the first phase. We have the following result for this scheme:
\begin{proposition} \label{prop:YamItoh}
	Let $P_q(y|x)$ be the channel between the query result and the query output for a query region of size $q$. The targeting rate-reliability tradeoff for non-adaptive scheme with a Yamamoto-Itoh validation is given by	
\begin{align}
	E = C_1(0)\rbr{1-\frac R {I_{XY}(q^*;q^*)}}
\end{align}
where $q^*$ is the maximizer in \eqref{eq:rate_non_adaptive}.
\end{proposition}
\textit{Proof}: Let $\varepsilon$ denote an error event. By the Chernoff-Stein lemma \cite[Section 11.8]{cover}, the probability of declaring $H_0$ when $H_1$ is true (false erasure, or ``false alarm'' in the Neyman-Pearson terminology) can be arbitrarily small while attaining 
\begin{align}
	\Pr(\varepsilon)\leq 2^{\lambda ND(P_{\delta}(y|x=1)||P_{\delta}(y|x=0))} = 2^{\lambda N C_1(\delta)}
\end{align}
Letting $\m{E}$ denote the erasure event and $\varepsilon_1$ denote the event of error in the first phase, we have that $\Pr(\m{E})\leq \Pr(\varepsilon_1) + \Pr(H_0|H_1)$, where $\Pr(H_0|H_1)$ denotes the probability that $H_0$ was declared although the target was hit in the first phase. The latter probability can be made arbitrarily small for large enough $N$ by the Chernoff-Stein lemma while the former probability vanishes exponentially fast with $N$ by Theorem \ref{thrm:non-adapt}. Therefore, for any $\epsilon > 0$ there exists $N$ large enough such that $\Pr(\m{E})\leq \epsilon$. Therefore,  $\bE[\tau]\leq \frac N {1-\epsilon}$.
We use the highest possible rate of $I_{XY}(q^*;q^*)=\frac{\log\frac 1 {\delta}}{(1-\lambda)N}$ for the first phase. Therefore, the effective rate of  one round (without taking erasure into account) is given by $R$, we have that $\lambda = 1-\frac R {C(q^*)}$. Using this, the error exponent attained by this scheme is lower bounded by
\begin{align}
	\frac {-\log \Pr(\varepsilon)} {\bE[\tau]} \geq C_1(\delta)\rbr{1-\frac R {I_{XY}(q^*, q^*)}}(1-\epsilon) 
\end{align}
Since both $\delta$ and $\epsilon$ can be made arbitrarily small, this scheme attains the exponent given in Proposition \ref{prop:YamItoh}

\begin{remark}\label{RemarkDoubleGain}
Note that the gain from adaptivity is twofold here. Not only does adaptivity allow us to validate the first phase result and increase the error exponent, it also allows us to validate the result over the best possible channel $P_{0}(y|x)$. Therefore the attained exponent is higher than the optimal Burnashev \cite{burnashev1974interval} exponent for the channel $P_{q^*}(y|x)$.  
\end{remark}

%

\subsection{Two-Phase Search with Validation}\label{Sec:2PhaseAndValidation}
In this section, we show that a two-phase scheme with Yamamoto-Itoh validation achieves the best possible performance, improving upon non-adaptive strategies (with and without validation) both in maximal targeting rate and in targeting rate-reliability tradeoff. 
      
\begin{theorem}\label{Thm:ThreePhase}
Let $P_q(y|x)$ be the query observation channel and let $q^*$ be the optimal solution in \eqref{eq:rate_non_adaptive}. For any $\alpha\in (0,\tfrac{1}{2})$, there exists a $\gamma\in(0,1)$ and search scheme with error probability $\eps$ and resolution $\delta$, satisfying   
\begin{equation*}
\label{binC}
\mathbb{E} [\tau]\le  \left(\frac{\log(1\slash\alpha)}{I_{XY}(q^*,q^*)} + \frac{\log (1\slash\delta)}{\max_q I_{XY}(q,\alpha q)} + \frac{\log(1\slash\epsilon)}{C_1(\delta)}\right)/(1-\gamma).
\end{equation*}
where $\gamma\to 0$  as $\alpha\to 0$ and $\delta \to 0$.
\end{theorem}
\begin{corollary}\label{cor:best}
By letting $\alpha$ vanish much slower than $\delta$, and noting that $\gamma$ can be arbitrarily small, we conclude that the maximal targeting rate for adaptive schemes is given by 
  \begin{equation*}
    C(0) = \max_{p\in[0,1]} I_{XY}(p,0),
  \end{equation*}
which is the capacity of the best observation channel associated with the measurements, hence the best possible. The associated targeting rate-reliability tradeoff is 
\begin{equation*}
  E(R) =  C_1(0)\left(1 -\frac{R}{C(0)}\right).
\end{equation*}
which is also the best possible.   
\end{corollary}
\begin{remark}
 Juxtaposing Theorem \ref{thrm:non-adapt} and Corollary \ref{cor:best} above, we conclude that in contrast to the case of constant measurement independent noise, adaptive search strategies outperform the optimal non-adaptive strategy in both targeting rate and reliability.  
\end{remark}

\begin{proof}
We prove the theorem for a fixed $\alpha$ and $\delta,\eps\to 0$. In the first search phase, the agent employs the optimal non-adaptive search strategy with $\tau=N_1$ and resolution $\alpha$. By Theorem \ref{thrm:non-adapt}, as long as $N_1 > \frac {\log \frac 1 {\alpha}}{I_{XY}(q^*,q^*)}$, then at the end of this phase the agent knows an interval of size $\alpha$ containing the target with probability $1- \epsilon_1$, where $\epsilon_1$ vanishes as $N_1$ grows. 

In the second phase, the agent ``zooms-in'' and performs the search only within the $\alpha$-sized interval obtained in the first phase. To that end, the agent employs the optimal non-adaptive search strategy with $\tau=N_2$ and resolution $\delta$, with the query sets properly shrunk by a factor of $\alpha$. We note that in this phase, all queried sets are of size smaller than $\alpha$. Therefore, after appropriately scaling the affect of the query size in Theorem \ref{thrm:non-adapt}, as long as $N_2 > \frac{\log 1\slash\delta}{\max_q I_{XY}(q,\alpha q)}$ and given that the first phase outcome was correct (the decided $\alpha$-sized region contained the target), then at the end of the second phase the agent knows an interval of size $\delta$ containing the target, with probability $1- \epsilon_2$, where $\epsilon_2$ vanishes to zero as $N_2$ grows with the random coding error exponent given by Theorem \ref{thrm:non-adapt}. 

At this point, the agent performs the Yamamoto-Itoh validation step of length $N_3$, which queries a fixed interval of size $\delta$. If not successful, the agent repeats the whole three-phase process from scratch. By the same arguments that lead to Proposition \ref{prop:YamItoh}, the probability of falsely declaring a $\delta$ region as containing the target when the target is not contained in it (namely, not detecting an error in the first or second phase) can be smaller than $\epsilon$ as long as $N_3>\frac{\log(1\slash\epsilon)}{C_1(\delta)}$.  
If the last phase declares an erasure, we start the three-phase search from scratch. Denote the erasure event by $\varepsilon$. The probability of $\varepsilon$ is upper bounded  by the sum of the first and second phase errors, which can be made arbitrarily small by letting $N_1,N_2$ grow. Therefore, the expected stopping time of this whole procedure is $\frac{N_1+N_2+N_3}{1-\Pr(\varepsilon)}$, which is the statement of the theorem with $\gamma=\Pr(\varepsilon)$. Since  by Theorem \ref{thrm:non-adapt} the first and second phases errors vanish as $N_1,N_2$ become large, $\varepsilon$ vanishes as $N_1+N_2$ becomes large (equivalently, as $\alpha,\delta$ become small), and the proof of Theorem \ref{Thm:ThreePhase} is concluded. 

To prove the corollary,  we let $\alpha$ vanish much slower than $\delta$, such that the length of the first phase becomes negligible compared to the other two phases. Letting $N_2=\lambda N$ and $N_3 = (1-\lambda)N$ and choosing $\lambda$ as in Proposition \ref{prop:YamItoh} we arrive the the Burnashev error exponent. As $\alpha\to 0$ our search channel for the second phase becomes arbitrarily close to the best possible channel we can operate over ($P_0(y|x)$), which by our assumptions on monotonicity of the noise level with the interval size, has the highest capacity. Therefore by the channel coding theorem, no higher targeting rate can be achieved. Indeed, $\lim_{\alpha\to 0}\max_q I_{XY}(q,\alpha q)=\max_q I_{XY}(q,0)=C(0)$. Furthermore, since we achieve the optimal error exponent for channels with feedback for $P_0(y|x)$, by Burnashev's result \cite{burnashev1974interval}, we achieve the best possible rate-reliability tradeoff. 

\end{proof}

\subsection{Numerical Results for Example \ref{Ex:BSC}}
The rate-reliability tradeoff of the above suggested adaptive schemes as well as a lower bound on the non adaptive one for our binary Example \ref{Ex:BSC} is depicted in Fig. \ref{Fig:Exponents}. The gain in terms of both rate and reliability of adaptive schemes is apparent. It is seen that Forney's decision feedback scheme (labeled  by 'b') which upper bounds the sphere packing bound and hence the exponent of any non-adaptive scheme is far below the exponents attained by the adaptive validation schemes. The twofold gain mentioned in Remark \ref{RemarkDoubleGain} can be seen in the difference between the lines labeled 'c' and 'd', where 'c' is the exponent attained by performing the Yamamoto-Itoh validation on $P_{q^*}(y|x)$ while 'd' is the performance attained by validation over the best possible channel. Finally, 'e' depicts the result of Theorem \ref{Thm:ThreePhase} as applied to Example \ref{Ex:BSC}, where we added a search phase and attained both the best possible exponent and the optimal targeting rate, which pertains to the capacity of the best possible observation channel.    

\begin{figure}[htp]
\centering
\includegraphics[width=0.45\textwidth]{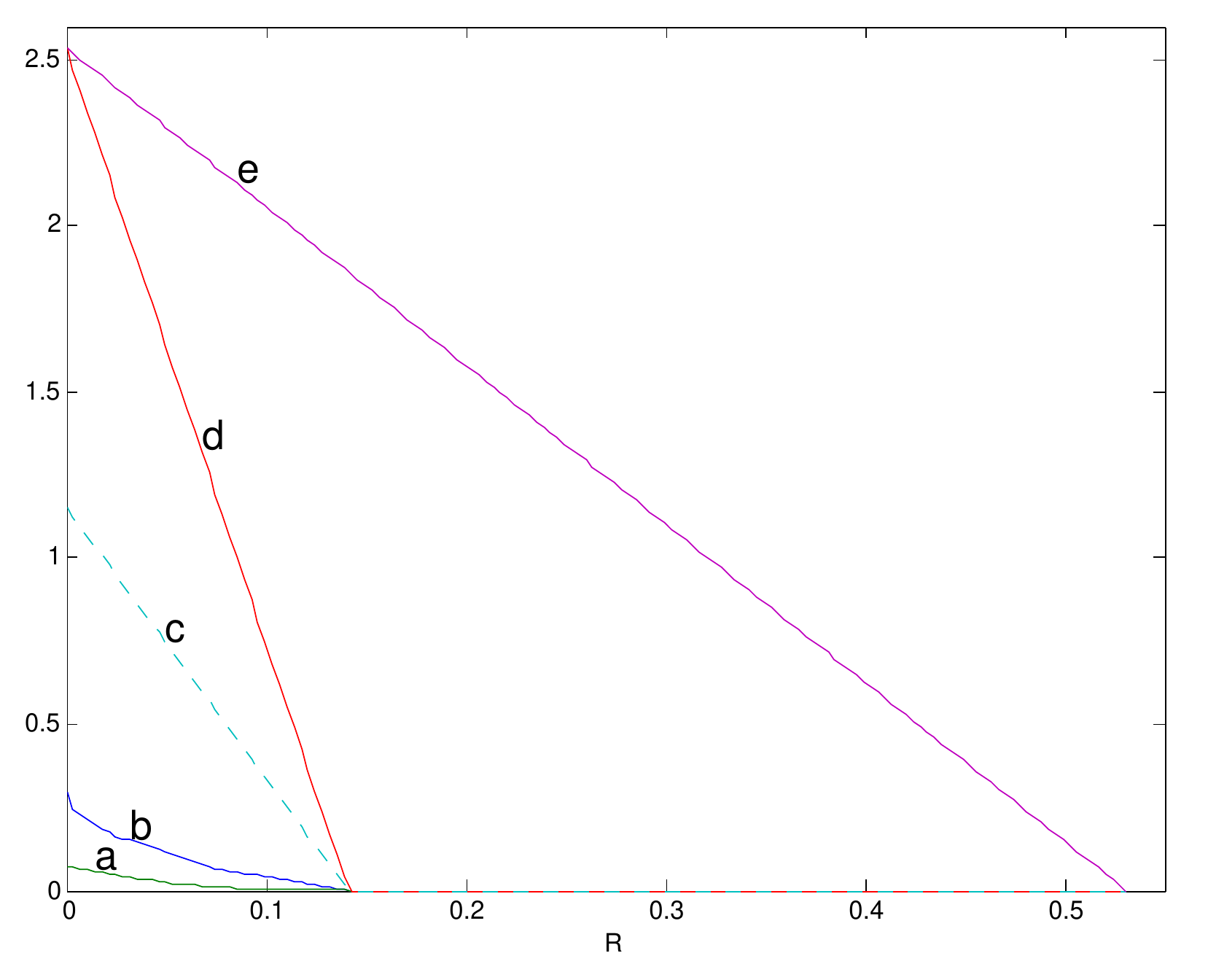} 
\caption{Error exponents (known velocity) for noise growing linearly with size: $p[0]=0.1, p[\frac 1 2]=0.45$ (a) Random coding (b) Decision feedback (c) Burnashev's upper bound for BSC$(p[q*])$ (d) Yamamoto-Itoh validation for the non-adaptive scheme (e) Yamamoto-Itoh validation for BSC$(p[0])$}.\label{Fig:Exponents}
\end{figure}

\section{Moving Target  with Unknown Velocity}\label{Sec:MovingT}
In this section, we allow the target to move at an unknown constant velocity. We will be interested in finding not only the initial location, but also the velocity, both up to some fixed resolution. As mentioned in the Introduction, adding this extra unknown parameter to the search problem should intuitively decrease the rate by a factor of two. This penalty stems from the fact that even if we could freeze the target at different times (and positions) and run our stationary target search, in order to deduce its velocity we would  need to run the search twice, and compute the velocity from the distance between the frozen positions. The results we present in this section indeed support this intuition.

Since we do not want our moving target to fall into oblivion and to forever be lost when it reaches the edges of the unit interval, in this section we consider targets moving on the unit circle. Clearly, with a sample rate of one sample per second, we will not be able to distinguish e.g. between a stationary target and a target that traverses at a velocity of exactly one full circle per second. For that reason, we consider only velocities $v$ such that $|v|\leq v_{\max}\leq \half$, where a negative velocity denotes a counterclockwise movement. While our results hold for any $v_{\max}\leq \half$, from a practical standpoint $v_{\max}$ should be a very small number; it is quite unreasonable to assume that the target moves across the entire search area within a small number of samples. 

Below we only extend the non-adaptive search to include moving targets. Extending  adaptive schemes follows along the same lines as before, where we use the non-adaptive scheme in several phases: in the first phase we narrow the search region, and narrow down both location and velocity in the second phase. To avoid repetition, we omit the details of this extension. 

\subsection{Setting and Preliminaries} 
We naturally identify the unit interval modulo $1$ with a circle of unit circumference. Correspondingly, the \textit{cyclic distance} between two points $w,w'$ in the unit interval is 
\begin{align*}
  |w-w'|_c \dfn \min\{|w-w'|, 1-|w-w'|\}
\end{align*}

The target starts at an arbitrary location, $w_0$ and at time $n$ its position is given by 
\begin{align*}
W_n = w_0 + V\cdot n \mod1
\end{align*}
where $V$ is uniformly distributed on $[-v_{\max}, v_{\max}]$ and $\vmax \leq \half$. Note that while we have no assumptions on the initial location, we do assume that the velocity of the target is uniformly distributed\footnote{we note that the dithering we used in the previous section can be adopted to this case to allow us to consider arbitrary velocity if $v_{\max}=\half$, however, this cannot be readily done for other values of $v_{\max}$. }. The noise and search models remain unchanged from the previous section. A non-adaptive strategy is comprised of $N$ search regions $\cbr{S_n}_{n=1}^N$ and two estimators $\hat{W}_N = \hat{W}_n(Y^N)$, $\hat{V} = \hat{V}(Y^N)$ where $Y^N$ is the vector of $N$ observations resulting from the $N$ queries defined by $\cbr{S_n}_{n=1}^N$. A strategy is said to have search resolution $\delta$ and error probability $\epsilon$ if for any $w_0$,
\begin{align}
	\Pr\left(\max\cbr{|\hat{W}_N - W_N|_c, |\hat{V}-V|} \leq \delta\right) \geq 1-\epsilon.
\end{align}
 
\subsection{Main result for unknown velocity}
We have the following theorem
\begin{theorem}
Let $P_{q}(y|x)$ be the query observation channel and let the target velocity be uniform on $[-v_{\max}, v_{\max}]$. Then any targeting rate $R$ satisfying 
\begin{align}
	R < \max_{q} \frac 1 2 I_{XY}(q,q)(1-2v_{\max}), 
\end{align}
is achievable with targeting reliability $\frac{E(R(1-2v_{\max}))}{1-2v_{\max}}$ using non-adaptive search strategies, where $E(R)$ is the random coding error exponent for the channel  $P_{q^*}(y|x)$. Furthermore, the maximal targeting rate for non-adaptive search strategies is at most $\max_{q} \frac 1 2 I_{XY}(q,q)$. 
\end{theorem}
\begin{remark}
	It is evident that the direct and converse parts of this theorem become tight as $v_{\max}$ becomes small. As pointed out above, $v_{\max}$ will be an extremely small number in any reasonable search scenario. The converse part is in line with the intuition that at least twice the measurements are needed to capture both location and velocity. The gap between the converse and direct parts stems from the fact that the paths of targets starting at different initial positions and velocities, can intersect several times during the sampling period. The most likely error to make is between two trajectories that intersect many times. While most trajectories will intersect no more than once, those with very different velocities can intersect many times during the $N$ queries. As will be shown below, the number of times trajectories can intersect is governed by $2v_{\max}$.	
\end{remark}

\begin{remark}
  Note that even as $v_{\max}\to 0$, we still incur a factor half in the rate relative to the case of a stationary target. This discontinuity stems from the order of limits; namely, even when the target is guaranteed to move at an extremely small velocity, it can still cover arbitrarily large distances given sufficient time. 
\end{remark}

\subsection{Proof of Achievability}
Each pair of initial position and velocity $(w_0,v)$ naturally induces a \textit{trajectory}:
\begin{align*}
\tau(w_0,v) \dfn \{w_0+vn \mod1\}_{n=0}^{N-1}.
\end{align*}
Set some desired resolution $\delta$. We partition the unit interval into $M = N/\delta$ equi-sized subintervals (sensors) $\{b_m\}$, of size $\delta/N$ each. The associated \textit{quantized trajectory} $m(w_0,v)$ (w.r.t. to this partition) is given by 
\begin{align*}
\hat{\tau}(w_0,v) \dfn \textrm{sensor}\left(\mu(w_0,v)\right)
\end{align*}
where $\textrm{sensor}(x) = m$ if $x\in b_m$. We say that two (quantized) trajectories with parameters $(w_0,v)$ and $(w_0',v')$ are \textit{close} if $|w_0-w_0'|_c\leq \delta$ and $|v-v'|\leq \frac{\delta}{N}$. Otherwise, we say the (quantized) trajectories are \textit{far}.  

\begin{lemma} \label{Lem:NumOfTraj}
	The number of distinct quantized trajectories is at most $(2Nv_{\max}+3)N^2M^2$. 
\end{lemma}
\begin{proof}
Let $\mathcal{T}$ be the set of all trajectories, and let $\mathcal{T}(m,m')\subset \mathcal{T}$ be the set of all trajectories $\tau$ that start in sensor  $m$ and end in sensor $m'$, i.e., $\textrm{sensor}(\tau(0)) = m$ and $\textrm{sensor}(\tau(N-1)) = m'$. Moreover, let $\mathcal{T}_c(m,m')$ be the set of all the trajectories in $\mathcal{T}(m,m')$ that start and end at exactly the centers of the sensors $m$ and $m'$. Let $\widehat{\mathcal{T}},\widehat{\mathcal{T}}(m,m')$ and $\widehat{\mathcal{T}}_c(m,m')$ be the corresponding sets of quantized trajectories. 

We are interested in upper bounding $|\widehat{\mathcal{T}}|$. To that end, we first write 
\begin{align*}
|\widehat{\mathcal{T}}| &= \sum_{m,m'\in[M]}|\widehat{\mathcal{T}}(m,m')| \\
&\leq M^2\cdot \max_{m,m'\in[M]} |\widehat{\mathcal{T}}(m,m')|\\
&\leq M^2\cdot \max_{m,m'\in[M]} |\widehat{\mathcal{T}}_c(m,m')| \cdot \max_{m,m'\in[M]}\frac{|\widehat{\mathcal{T}}(m,m')|}{|\widehat{\mathcal{T}}_c(m,m')|}\\
&\leq (2Nv_{\max}+3)M^2\cdot \max_{m,m'\in[M]}\frac{|\widehat{\mathcal{T}}(m,m')|}{|\widehat{\mathcal{T}}_c(m,m')|}.
\end{align*}
For the last inequality we observe the following: The distance covered by a target that starts at location $w$ and ends up at location $w'$ is $|w'-w|_c + k$ for some integer $k$, hence the velocity of the target must be $v = \frac{|w'-w|_c}{N} + \frac{k}{N}$. Since $|v|\leq v_{\max}$ it must hold that $|k| \leq Nv_{\max} +1$. 

To conclude the proof, we now show that $|\widehat{\mathcal{T}}(m,m')|\leq N^2\cdot |\widehat{\mathcal{T}}_c(m,m')|$. Let $\hat{\tau}\in \widehat{\mathcal{T}}(m,m')$, and let $(w,v)$ be its initial position and velocity. It is easy to see that there must exist $\hat{\tau}_c\in \widehat{\mathcal{T}}_c(m,m')$ with velocity $v_c$ such that $|v-v_c|\leq \frac{1}{2MN}$. Thus, both quantized trajectories are always at a distance of at most $\frac{1}{M}$ (one sensor) from each other. Moreover, since both start and finish at the same two sensors, it must be that either $\hat{\tau} = \hat{\tau}_c$, or that $\hat{\tau}$ diverges from the $\hat{\tau}_c$ at some time instance and merges back to at another time instance. For any fixed $\hat{\tau}$. the number of such quantized trajectories $\hat{\tau}$ is at most $1+2{N\choose 2} \leq N^2$, where the factor $2$ corresponds to the fact that $\hat{\tau}$ can diverge either forward (in case it is faster) or backward (in case it is slower).

\end{proof}

\begin{lemma}\label{Lem:TrajIntersect}
	If two quantized trajectories are far, then they intersect at most $\lceil 2Nv_{\max}\rceil$ times. 
\end{lemma}
\begin{proof}
Consider two quantized trajectories with parameters $(w_0,v)$ and $(w_0',v')$, and assume they are far from each other. There are two cases: either $|v-v'|>\delta/N$, or both $|v-v'|\leq \delta/N$ and $|w_0-w_0'|_c > \delta$. The second case means that the trajectories start at least $N$ sensors apart, and since their absolute relative velocity is at most $\delta/N$ (a sensor per second), they can intersect at most once. Let us now consider the first case. Since here the absolute relative velocity is at least one sensor per second, then once the quantized trajectories intersect, their next intersection can occur only once the relative distance covered by the trajectories is at least one full circle. Since the total relative distance covered is at most $2v_{\max}N$, the maximal number of intersections is at most $\lceil 2v_{\max}N\rceil$. 
\end{proof}

Now, we proceed to prove the achievability via random coding using an input distribution $q^*$ that achieves the supremum in \eqref{eq:rate_non_adaptive}. As in the zero velocity case, we define a binary matrix, with $M$ rows and $N$ columns and draw its elements $B_{i,j}$ $i.i.d$ with $\Pr(B_{i,j}=1)=q^*$. Each row represents a specific sensor on the unit circle, where the first and last rows represent two adjacent sensors. In contrast to the zero velocity case, not only the rows are defined as codewords but also any path from left to right along this matrix is a possible codeword. In fact, each quantized trajectory $\hat{\tau}(w_0,v)$ naturally defines a diagonal line within the matrix that starts at the sensor containing $w_0$ and whose ``slope'' is determined by $v$. We can think of the bits visited along the way as a codeword. Collecting all those codewords generates a codebook $\cbr{\bx_k}_{k=1}^{\widetilde{M}}$, where $\widetilde{M} \leq 2N^2M^2$ by lemma  \ref{Lem:NumOfTraj}. Unlike the standard analysis of random coding, since trajectories overlap, their associated codewords are not independent across the overlaps, which complicates the analysis. On the bright side however, note that for a given correct codeword (corresponding to the true quantized target trajectory), the decoder can choose any one of the codewords that correspond to any of trajectory that is close to the correct one (where closeness is as defined above), and that will be considered a correct decoding. More specifically, for the correct codeword $\bx_k$, let $\calT_k$ be the set of all codewords that correspond to trajectories that are far from $\bx_k$. The decoder will err if and only if it decides in favor of some codeword from  $\calT_k$. 

To keep things simple, we employ a maximum-likelihood decision rule that selects the codeword $\bx$ that maximizes $P(\by|\bx)$ (with ties broken arbitrarily). Thus, a decoding error will occur only if $\exists x_{k'}\in\calT_k$ such that $P(\by|\bx_{k'})\geq P(\by|\bx_k)$, where $\bx_k$ is the correct codeword/trajectory. Note that due to the dependencies between the codewords, and since making an error that is not far from the truth is considered okay, this decoding rule is not necessarily the optimal one. 

As before, we apply Lemma \ref{Lem:DEFvanish} and the event $\m{A}$ to allow us to analyze the decoder for the channel $P_{q^*}(y|x)$. 
Following the same steps that led us to \eqref{eq:ExpInit}, we need to analyse $\Pr(e|\bx_k,\by,\m{A})$. We have
\begin{align*}
 	\Pr(e|\bx_k,\by,\m{A}) \leq \rbr{\sum_{k'\in \calT_k}\Pr(\m{E}_{k'}|\m{A})}^{\rho}
\end{align*}
where $\m{E}_{k'}$ is the event of erroneously declaring $\bx_k'$ when $\bx_k$ is the correct trajectory and
\begin{align}
 	\Pr(\m{E}_{k'}|\m{A}) = \sum_{\bx_{k'}\in\m{T}_k: P_{q^*}(\by|\bx_k)\leq  P_{\m{A}}(\by|\bx_{k'})} \Pr(\bx_{k'}|\bx_k,\m{A}) \label{eq:Moving1}
\end{align}

Repeating the steps that led to \eqref{eq:CondisUncond} we obtain
\begin{align}
	\Pr(\bx_{k'}|\bx_k,\m{A}) &= \frac{\Pr(\bx_{k'},\bx_k|\m{A})}{\Pr(\bx_{k}|\m{A})}\nl
	&\leq \frac {P(\bx_k,\bx_{k'})} {Q(\bx_{k})(1-\Pr(\m{A}^c)/Q_{min}^N)^2},
\end{align}
however, now  $\bx_k$ and $\bx_{k'}$ are independent only in the parts where they do not overlap. Using Lemma \ref{Lem:DEFvanish} and Taylor's expansion there exist an $\alpha_N$ such that $\frac 1 {(1-\Pr(\m{A}^c)/Q_{min}^N)^2} = 1-\alpha_N$ with $\alpha_N$ vanishing double exponentially in $N$. Substituting this into \eqref{eq:Moving1} and using Gallager's \cite{GallagerBook} method we arrive at:
\begin{align}
	\Pr(\m{E}_{k'}|\m{A}) \leq (1-\alpha_N)\sum_{\bx_{k'}\in\m{T}_k}  \frac {P(\bx_k,\bx_{k'})} {Q(\bx_{k})} \rbr{\frac{P_{q^*}(\by|\bx_{k'})}{P_{q^*}(\by|\bx_k)}}^{\frac 1 {1+\rho}}\label{eq:b4BreakingSum}
\end{align}
Now, using Lemma \ref{Lem:TrajIntersect}, we know that all $\bx_{k'}\in \m{T}_k$ intersect with $\bx_k$ at no more than $2Nv_{\max}$ indices. 
Let $\m{T}_k(d)$ denote the subset of trajectories in $\m{T}_k$ that contains all trajectories with exactly $d$ intersections with $\bx_k$, hence $\m{T}_k = \sum_{d=0}^{2Nv_{\max}} \m{T}_k(d)$. Using this, we rewrite \eqref{eq:b4BreakingSum} as
\begin{align}
	\Pr(\m{E}_{k'}|\m{A}) \leq (1-\alpha_N)\sum_{d=0}^{2Nv_{\max}}\sum_{\bx_{k'}\in\m{T}_k(d)}  \frac {P(\bx_k,\bx_{k'})} {Q(\bx_{k})} \rbr{\frac{P_{q^*}(\by|\bx_{k'})}{P_{q^*}(\by|\bx_k)}}^{\frac 1 {1+\rho}}
\end{align}
For a given $d$, let the set of indices where $\bx_k$ and $\bx_{k'}$ do not intersect be denoted by $\m{I}_d$. We know that there are exactly $N-d$ indices in this set. We have that\footnote{this is an inequality since the left-hand-side can be zero if the codewords do not agree on the coordinates where they overlap.} $P(\bx_{k},\bx_{k'}) \leq  Q(\bx_k)\prod_{i\in\m{I}_d}Q(x_{k',i})$. In addition, using the memorylessness of the channel and since $P_{q^*}(y_i|x_{k,i})=P_{q^*}(y_i|x_{k',i})$ on all indices where the trajectories corresponding to the codewords intersect, we have
\begin{align}
	\Pr(\m{E}_{k'}|\m{A}) \leq (1-\alpha_N)\sum_{d=0}^{2Nv_{\max}}\sum_{\bx_{k'}\in\m{T}_k(d)}  \prod_{i\in \m{I}_d} Q(x_{k',i}) \rbr{\frac{P_{q^*}(y_i|x_{k',i})}{P_{q^*}(y_i|x_{k,i})}}^{\frac 1 {1+\rho}}
\end{align}
After substituting the above equation in \eqref{eq:ExpInit} we get
\begin{align}
& \Pr(e|\m{A})\nl
 & \leq \widetilde{M}^{\rho}(1-\alpha_N) \int_{\by}d\by\sum_{\bx_k}\Pr(\bx_k|\m{A})P_{q^*}(\by|\bx_k)\sbr{\sum_{d=0}^{2Nv_{\max}}\sum_{\bx_{k'}: k'\in\m{T}_k(d)}  \prod_{i\in \m{I}_d} Q(x_{k',i}) \rbr{\frac{P_{q^*}(y_i|x_{k',i})}{P_{q^*}(y_i|x_{k,i})}}^{\frac 1 {1+\rho}}}^{\rho}\nl
 &= \widetilde{M}^{\rho}\gamma_N \int_{\by}d\by\sum_{\bx_k}Q(\bx_k)P_{q^*}(\by|\bx_k)\sbr{\sum_{d=0}^{2Nv_{\max}}\sum_{\bx_{k'}: k'\in\m{T}_k(d)}  \prod_{i\in \m{I}_d} Q(x_{k',i}) \rbr{\frac{P_{q^*}(y_i|x_{k',i})}{P_{q^*}(y_i|x_{k ,i})}}^{\frac 1 {1+\rho}}}^{\rho}\nl
 &\leq \widetilde{M}^{\rho}\gamma_N \int_{\by}d\by\sum_{\bx_k}Q(\bx_k)P_{q^*}(\by|\bx_k)\sum_{d=0}^{2Nv_{\max}}\sbr{\sum_{\bx_{k'}: k'\in\m{T}_k(d)}  \prod_{i\in \m{I}_d} Q(x_{k',i}) \rbr{\frac{P_{q^*}(y_i|x_{k',i})}{P_{q^*}(y_i|x_{k ,i})}}^{\frac 1 {1+\rho}}}^{\rho}\label{eq:Gallieq}\\
 &=  \widetilde{M}^{\rho}\gamma_N \sum_{d=0}^{2Nv_{\max}} \int_{\by}d\by\sum_{\bx_k}\Pr(\bx_k)P_{q^*}(\by|\bx_k)\sbr{\sum_{\bx_{k'}\in\m{T}_k(d)}  \prod_{i\in \m{I}_d} Q(x_{k',i}) \rbr{\frac{P_{q^*}(y_i|x_{k',i})}{P_{q^*}(y_i|x_{k,i})}}^{\frac 1 {1+\rho}}}^{\rho}\nl
 &=  \widetilde{M}^{\rho}\gamma_N \sum_{d=0}^{2Nv_{\max}} \int_{\by}d\by\sum_{\bx_k}\Pr(\bx_k)P_{q^*}(\by|\bx_k)\sbr{\sum_{\bx_{k'}\in\m{T}_k(d)}  \prod_{i=1}^{N-d}Q(x)\prod_{i\in \m{I}_d}  \rbr{\frac{P_{q^*}(y_i|x)}{P_{q^*}(y_i|x_{k,i})}}^{\frac 1 {1+\rho}}}^{\rho}
\end{align}
where $\gamma_N\eqd 1-\alpha_N$ converges to unity double exponentially fast and where \eqref{eq:Gallieq} holds since $\sbr{\sum_ia_i}^{\rho} \leq \sum_i a_i^{\rho}$ for any $0\leq \rho \leq 1$ and $a_i\geq 0$. Note that for any $d$ and any choice of $\bx$, while the bracketed term contains summation over vectors of length $N$, the product terms are over only $N-d$ coordinates, hence effectively we can think of the vectors as being of length $N-d$. Since we iterate through all possible $2^N$ binary vectors and integrate over all possible output sequences, after some standard Gallager-type manipulations we have that  
\begin{align}
\Pr(e|\m{A})& \leq  \widetilde{M}^{\rho}\sum_{d=0}^{2Nv_{\max}}\gamma_N \prod_{i=1}^{N-d}\int_{y}dy\sbr{\sum_{x}  Q(x)P^{\frac 1 {1+\rho}}_{q^*}(y|x)}^{1+\rho} \nl
&= \sum_{d=0}^{2Nv_{\max}} 2^{-(N-d)(E_0(\rho) - \rho \log( \widetilde{M})/(n-d)) + \log\gamma_N} 
\end{align}
Since the number of elements in the sum grows linearly with $N$, the dominating  element will be the one with the largest exponent, which is given by 
\begin{align}
 	2^{-N(1-2v_{\max}))(E_0(\rho) - \rho \log( \widetilde{M})/n(1-2v_{max})) + \log\gamma_N} 
\end{align}

\subsection{Proof of Converse}
The converse follows the same lines as the converse for stationary targets, only now there are two unknowns: initial position and velocity.
Denote the fixed stopping time by $\tau = N$. Let  $\{S_n\}_{n=1}^N$ be any non-adaptive strategy achieving an error probability $\eps$ with search resolution $\delta$. Partition the unit interval into $\lceil\beta/\delta\rceil$ equi-sized intervals for some constant $\beta\in(0,\tfrac{1}{2})$, and let $W_N'$ be the index of the interval containing $W_N$. Similarly, partition $[-v_{\max, v_{\max}}]$ into $\lceil2v_{\max}\beta/\delta\rceil$ equi-sized intervals and let $V'$ be the index of the interval containing $V$. It is easy to see that the scheme $\{S_n\}$ can be tweaked to output $W_N', V'$ with error probability at most $\eps'\dfn \eps + 4\beta(1-\beta)$, where the latter addend stems from the probability that $(\wh{W}_N,\wh{V})$ is too close to a boundary point. 

Note that $X_n\sim\mathrm{Bern}(q_n)$ where $q_n\dfn |S_n|$, and that $Y_n$ is obtained from $X_n$ through a memoryless time-varying channel $P_{q_n}(y|x)$. Following the steps of the converse to the channel coding theorem, we have
\allowdisplaybreaks{\begin{IEEEeqnarray}{rCl}\label{eq:converse}
2\log\left(\frac{\beta}{\delta}\right) +\log(2v_{\max})&=& H(W_N',V') \notag\\ 
&=& I(W_N',V';Y^N) + H(W_N',V'|Y^N) \notag\\ 
&\stackrel{(\textrm{a})}{\leq}& I(W_N',V';Y^N) + N\eps' \notag\\ 
&=& \sum_{n=1}^N I(W_N',V';Y_n|Y^{n-1}) + N\eps' \notag\\ 
&\leq& \sum_{n=1}^N I(W_N',V',Y^{n-1};Y_n) + N\eps' \notag\\ 
&\leq& \sum_{n=1}^N I(W_N',V',W_n,Y^{n-1};Y_n) + N\eps' \notag\\ 
&\stackrel{(\textrm{b})}{=}& \sum_{n=1}^N I(X_n;Y_n) + N\eps' \notag\\ 
&\stackrel{(\textrm{c})}{=}& \sum_{n=1}^N I_{XY}(q_n,q_n) + N\eps', 
\end{IEEEeqnarray}}
where (a) is by virtue of Fano's inequality, (b) follows since $X_n$ is a function of $W_n$ and the measurement noise is independent across time, and (c) stems from the fact that the channels $\cbr{P_{q_n}(y|x)}_{n=1}^N$, while time varying, are a fixed function of the codebook. Dividing by $N$ we obtain 
\begin{IEEEeqnarray*}{rCl}
R &=& \frac{\log(1/\delta)}{N} \leq  \frac{1}{2N}\left(\sum_{n=1}^N I_{XY}(q_n,q_n) - 2\log\beta -\log(2v_{\max})\right) + \frac{\eps'}{2} \\
&\leq&  \frac{1}{2}\left(\sup_{q\in(0,\frac{1}{2})} I_{XY}(q,q) -\frac{2\log\beta-\log(2v_{\max})}{N}+ \eps + 4\beta(1-\beta)\right).
\end{IEEEeqnarray*}
Noting that the inequality above holds for any $\beta\in(0,\tfrac 1 2)$, the converse now follows by taking the limit $N\to\infty$, and then requiring $\eps\to 0$. 


\section{Conclusions and Further Research} \label{Sec:Conclusion}
In this paper, we considered the problem of acquiring a target moving with known/unknown velocity on a circle, starting from an unknown position, under the physically motivated observation model where the noise intensity increases with the size of the queried region. For a known velocity, we showed that unlike the constant noise model, there can be a large gap in performance (both in targeting rate and reliability) between adaptive and non-adaptive search strategies.  Furthermore, we demonstrated that the cost of accommodating an unknown velocity in the non-adaptive setting, incurs a penalty factor of at least two in the targeting rate.  

One may also consider other search performance criteria, e.g., where the agent is cumulatively penalized by the size of either the queried region or its complement, according to the one containing the target. The rate-optimal scheme presented herein, which is based on a two-phase random search, may be far from optimal in this setup. In such cases we expect that sequential search strategies, e..g, ones based on posterior matching \cite{Shayevitz11,naghshvar2013extrinsic}, would exhibit superior performance as they naturally shrink the queried region with time. Other research directions include more complex stochastic motion models, as well as searching for multiple targets (a ``multi-user'' setting). For the latter, preliminary results indicate that the gain reaped by using adaptive strategies vs. non-adaptive ones diminishes as the number of targets increases \cite{MeGT2015}.

\bibliographystyle{IEEEtran}
\bibliography{Bib-1}

\end{document}